\documentclass[12pt, draftclsnofoot, onecolumn]{IEEEtran}
\usepackage{graphicx}
\usepackage{amssymb}
\usepackage{amsbsy}
\usepackage{amsmath}
\usepackage{amsthm} 
\usepackage{comment}
\usepackage{bbm}
\usepackage{cite}
\usepackage{caption}
\usepackage{subcaption}
\usepackage{color}
\usepackage{cleveref}
\newtheorem{thm}{Theorem}
\newtheorem{cor}{Corollary}
\newtheorem{lem}{Lemma}
\newtheorem{remark}{Remark}
\usepackage{cite}

\ifCLASSINFOpdf

\else

\fi

\begin{document}
\title{Backscatter Communications for the Internet of Things: A Stochastic Geometry Approach}
\author{\IEEEauthorblockN{Mudasar Bacha and Bruno Clerckx}\thanks{M. Bacha and B. Clerckx is with the Communication and Signal Processing Group, Department
of Electrical and Electronic Engineering, Imperial College London, London
SW7 2AZ, U.K. (e-mail: {m.bacha13, b.clerckx}@imperial.ac.uk). } 
\thanks{This work was supported in part by the EPSRC of UK, under Grant EP/P003885/1.}}
\maketitle
\begin{abstract}
Motivated by the recent advances in the Internet of Things (IoT) and in Wireless Power Transfer (WPT), we study a network architecture that consists of power beacons (PBs) and passive backscatter nodes (BNs).  The PBs transmit a sinusoidal continuous wave (CW) and the BNs reflect back a portion of this signal while harvesting the remaining part. A BN harvests energy from multiple nearby PBs and modulates its information bits on the composite CW through backscatter modulation. The analysis poses real challenges due to the double fading channel, and its dependence on the PPPs of both the BNs and PBs. However, with the help of stochastic geometry, we derive the coverage probability and the capacity of the network in tractable and easily computable expressions, which depend on different system parameters.  We observe that the coverage probability decreases with an increase in the density of the BNs, while the capacity of the network improves. We further compare the performance of this network with a regular powered network in which the BNs have a reliable power source and show that for a very high density of the PBs, the coverage probability of the former network approaches that of the regular powered network.
\end{abstract}
\section{Introduction}
The emerging Internet of Things (IoT) is expected to connect billions of small computing devices to the Internet \cite{cisco}. These tiny devices have  processing, sensing and wireless communications capabilities. They are supposed to be deployed everywhere  and can be accessed from anywhere at any time. However, powering these small devices is one of the main challenges (others are interoperability, management, security and privacy) \cite{challenges_1, challenges_2}. It is very expensive and impractical to replace the batteries of such a massive number of devices or power them with wires. Therefore, harvesting energy from external sources for perpetual operation is a viable option \cite{challenges_2}. 
The advancement in  Wireless Power Transfer (WPT) has made it possible to power IoT devices \cite{bruno_WPT}.  

We study a random network consisting of passive backscatter nodes (BNs) and  power beacons (PBs). 
The PBs are deployed for WPT  and they transmit sinusoidal continuous wave (CW)\footnote{Significant enhancements can nevertheless be obtained by proper multisine waveform designs \cite{ekaterina,zati}.}. The mobile nodes or backscatter nodes (BN), which are in the range of PBs will backscatter a  portion of this CW to a nearby receiver by mismatching its antenna impedance, while harvesting the remaining power to operate its integrated circuit \cite{boyer}. A BN modulates the data onto the backscattered signal by controlling its antenna impedance. A schematic diagram of backscatter communication system is shown in Fig. \ref{BN_node}. In contrast to conventional radio architecture relying on power hungry RF chains, the backscatter nodes do not have any active RF components. As a result, the BN has miniature hardware and very low power consumption. For more details on the backscatter operation, the readers are referred to \cite{boyer, kaibin_magzine}, and the references therein.  The aim of this work is to study the performance of this unique network, where the BNs are powered through wireless power transfer from the PBs. The analysis is performed by using tools from the stochastic geometry. 
	\begin{figure} 
	\centering
		\includegraphics[scale=0.4]{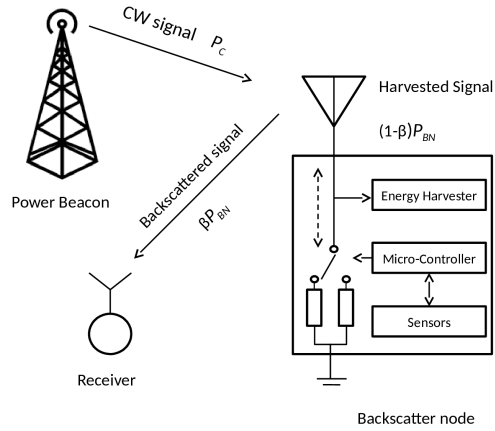}
		\caption{Backscattering communication system}
	\label{BN_node}
	\end{figure}
\begin{table*}
\caption{List of Notations}
	\centering
		\begin{tabular}{|c|l|}
		\hline
		Notation & Description \\
		\hline 	\hline	
		$\Phi_p$, $\Phi_b$, $\lambda_p$, $\lambda_b$ & PPP of the PBs, PPP of the BNs, density of the PBs, density of the BNs   \\	\hline
		$P_C$, $\beta$, $N_p$ & transmit power of each PB, backscattering efficiency, average number of the PBs in the harvesting region\\ \hline
		$\alpha^f$, $\alpha^b$ & path loss exponent of the forward channel, path loss exponent of the backward channel \\ \hline
		$P_{BN}$, $\bar{P}_{BN}$ & received power at BN, average received power at BN \\ 	\hline
		$d_{i,j}$, $x_{j,k}$ & distance between $i$th receiver and $j$th BN, distance between $j$th BN and $k$th PB \\ \hline
		$\mathrm{Ei}\left(z\right)$ & $\mathrm{Ei}\left(z\right)$ is the exponential integral function $\mathrm{Ei}\left(z\right)=\int_z^{\infty}\mathrm{e}^{-t}t^{-1}\mathrm{d}t$  \\ \hline
		$\mathbb{P}_s, \mathbb{P}\left[.\right], \mathbb{E}_x$ & coverage probability, probability of an event, expectation with respect to $x$\\ \hline
		$\Theta_m, N_0, \gamma_{R_0}', \gamma_{R_0}$ &  SINR threshold, variance of the noise, SINR at typical receiver $R_0$, modified SINR at $R_0$ \\ \hline
		$h_{i,j}^f$ & complex channel gain of the forward channel between the $i$th BN and $j$th PB while $\mathbf{h}_{i,j}^f=\left|h_{i,j}^f\right|^2$ \\ \hline
		$h_{j,k}^b$ & complex channel gain of the backward channel between the $j$th receiver and $k$th BN while $\textbf{h}_{j,k}^b =\left|h_{j,k}^b\right|^2$ \\ \hline
		\end{tabular}
	\label{table1}
\end{table*}
\subsection{Related Work}
Backscatter communication is traditionally used in the Radio Frequency IDentification (RFID) systems, where a RFID-reader reads data from nearby tag via backscatter modulation\cite{boyer,kaibin_magzine,RFID_intro}.  Recently, there have been some studies, which use backscatter communication by harvesting energy from ambient sources such as TV and Wi-Fi signals \cite{ambient,wifi,modulation_air_ofdm}. However, this ambient energy harvesting technique cannot be used in the large scale IoT due to scalability issue \cite{kaibin_magzine, kaifeng}. The authors in \cite{collision_resol_letter} consider a single cell network in which there is one reader in the center of the cell and the sensors node are distributed uniformly within this cell. They studied the decoding probability under different collisions resolution schemes such as using directional antennas, successive interference cancellation, and ultra-narrow band transmission. They observed that a combination of these technique give significant gains. A full-duplex backscatter communication setup is considered in \cite{full_duplex_time_hoping} where a multiple access scheme based on time-hopping spread-spectrum is proposed to reduce the interference. A comparison between multistatic and monostatic radio architecture has been done in \cite{multistatic}, where it has been demonstrated that multistatic architecture outperforms the monostatic architecture in term of diversity order, bit error rate, energy outage and coverage. 

Most of the above works considered a single-reader single-tag or single-reader multiple-tags except \cite{kaifeng}, which consider a wireless powered backscatter communication network in which dense PBs are deployed to wirelessly power BNs. They  model the network as Poisson cluster process (PCP), where PBs constitute the parent point process and the BNs form daughter point processes. Using tools from stochastic geometry they studied the coverage and network capacity of the network. However, they considered that a BN only harvests energy from a single PB and they also did not consider fading in the forward channel. Moreover, they derived lower bounds on the coverage and capacity of the network. 

In our work, we study a more generic and close to practical network setup. In our setup both the PBs and the BNs are distributed according to PPPs and we consider fading in both the forward and backward channel. We also consider that the BNs can harvest and backscatter energy from multiple nearby PBs. Then we utilize tools from stochastic geometry to study the coverage and capacity of such network. 
Stochastic geometry has emerged as a powerful tool for the analysis of wireless networks due to its amazing tractability and accuracy \cite{hesham_tutorial}. It has been used in the analysis of the uplink cellular networks \cite{dude_mudasar}, the downlink cellular networks \cite{marco_dl},  ad-hoc networks \cite{jeff_adhoc}, and millimeter wave networks \cite{marco_milimeter}. Recently, it has been used to study the wireless powered network \cite{kaifeng,collision_resol_letter,joint_harpeet}. 



\subsection{Contributions and Outcomes}
We develop a comprehensive model for the wireless powered backscatter IoT network in which the locations of PBs and BNs are distributed as independent PPPs. The BNs modulate the information  on the unmodulated CW signal from the nearest PBs and reflect a portion of the signal to the receiver. The remaining energy is harvested for the operation of its integrated circuit. We study the coverage and capacity of this network for a typical BN where the capacity is defined as the number of successful transmissions between the BNs and their receivers in per unit area of the network.  In classical downlink cellular network, the coverage (and capacity) depends on a single channel (between a base station  and a mobile terminal) and the density of the base stations (BSs), while in the uplink it depends on the channel and the density of the mobile terminals (MTs). Furthermore,  the BS has a constant transmit power in downlink while the mobile terminals use power control in the uplink.  Whereas, in our setup, the coverage probability (and capacity) depends on the double fading channel, the PPP of BNs and the joint probability distribution function (PDF) of the distances of the nearest PBs. In addition, the reflected power of a BN is a random quantity and depends on the forward channel and the PPP of the PBs. Thus all these unique features make the analysis very challenging. 

To tackle those challenges, we use dimensionality reduction technique in which we apply the expectation with respect to the forward channel to both the signal and interference power before the analysis is carried out. This makes the analysis tractable and we find insightful  and accurate enough expressions for both the coverage and capacity of the network. Our analysis captures the effect of instantaneous energy harvesting from multiple PBs, backscatter modulations of the composite signal and double fading channel in a single tractable expression. This significantly contrasts with \cite{kaifeng} that uses a constant transmit power for the BNs and derive only lower bound on the coverage probability and capacity of the network. 

To get further insight, our coverage probability expression can be simplified into a single integral form when the BN harvests energy from a single nearest PBs and the path-loss exponent of both the forward and backward channel is the same. It can be even simplified to a closed form solution when we consider that the BNs reflect the signal with the mean harvested power. We further consider a network scenario in which the BN harvests from all the PBs in the network and backscatter their signal. This case is even more challenging because the coverage probability depends on the PPPs of both PBs and BNs, and the double fading channel, therefore, we only find the approximation for this case. The accuracy of the analysis is verified with extensive Monte Carlo simulations. 

Leveraging the above analysis, we are in a position to ask and answer some important and new questions/trade-offs that emerge in the wireless powered backscattered IoT network. Firstly, increasing the density of the PBs improves the harvested and reflected power, but at the same time it increases the amount of  interference. So, how does the coverage of the network vary with the density of the PBs? Secondly, what is the impact of the density of the BNs on the coverage and the capacity of the network? Finally, can the wireless powered backscattered network give the same coverage as of the classical cellular network where BNs are powered with reliable power source rather than through backscattering?
The answers to these questions are outlined below:
\begin{itemize}
	\item{}	We observe that increasing the density of the PBs increases the coverage probability of the BNs, where the coverage probability is defined as the probability that the instantaneous SINR at the randomly chosen receiver is greater than some predefined SINR threshold.  However, when the SINR threshold is high,  further increasing the density can decrease the coverage probability. This decrease in the coverage probability at a high SINR threshold is due to the increase in the interference power. The high SINR is achieved only by a very small fraction of the BNs, which are closer to the PBs and an increase in the density of the PBs increases both the signal and interfering power of the BNs. However, the increase in the interference power is more prominent than the increase in the signal power.
	\item{}  The coverage probability can be considered as the individual link reliability while the capacity is the number of successful transmissions. Increasing the density of the BNs decreases the coverage probability, whereas it increases the capacity of the network. 
	Therefore, the network provider should take into account the tradeoff between the individual link-reliability and the number of successful transmissions. 
	\item{} We compare the wireless powered backscatter IoT network with the regular powered network in which the BNs are supplied with reliable power source having the same transmit power as of the PBs. We observe that the coverage of the wireless powered backscatter network approaches to the coverage of the regular powered network for a very high density of the PBs. We further observe that up to a certain extent, increasing the density of PBs gives rapid increase in the coverage probability. However, beyond a certain density, any further  increase in the density of the PBs shows a very small improvement in the coverage probability. This suggests that it may not be cost-effective to keep increasing the density of PBs beyond a certain level.
\end{itemize}
	
The rest of the paper is organized as follows. In Section II, we present our system model and discuss our performance metrics. In Section III, we conduct the analysis and find the coverage and capacity of the network, while in Section IV we present the numerical and simulations results. Finally, we conclude and provide future research directions in Section V.

The key notations used in this paper are given in Table 1.

\section{System Model}
\subsection{Network Model}
\begin{figure*} 
	\centering
		\includegraphics[scale=0.28]{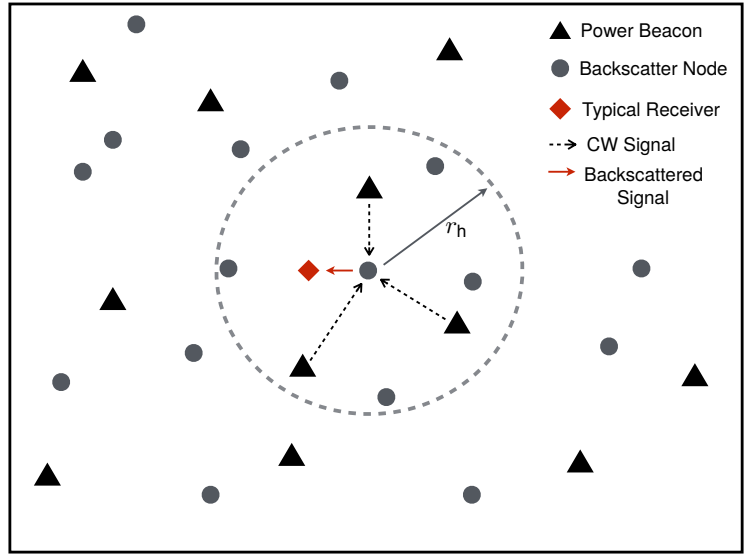}
		\caption{Network model}
	\label{Network_model}
	\end{figure*}
We consider an ad-hoc network consisting of PBs and BNs. The positions of both the PBs and BNs are modeled as independent homogeneous Poisson point processes (PPPs). Let $\Phi_p$ and $\Phi_b$ be the 2-D PPPs with densities $\lambda_p$ and $\lambda_b$ of the PBs and BNs, respectively. 
The PBs transmit a CW isotropically and all the PBs use the same frequency. A BN harvests some energy from this carrier\footnote{through AC-to-DC conversion} and modulates its data on the carrier through backscatter modulation \cite{boyer}. A BN receives a CW signal from multiple PBs and backscatters the signal. This reflected signal from the BN is received by a receiver, which can be located at a fixed distance from this BN in an isotropic direction. The network topology is shown in Fig. \ref{Network_model}. 
Without loss of generality, the analysis is performed for a \textit{typical} $\text{BN}_0$, which is transmitting the signal to its dedicated receiver $\text{R}_0$ \cite{martin_book}. 
In order to make the analysis tractable, we do not consider the circuit-power constraint and assume that the BNs modulate and  backscatter the signal all the time. Moreover, we consider that the PBs and the BNs have a single antenna. 

We assume that, around each BN, there is a region $z(Y,r_h)\subset\mathbb{R}^2$, which represents a disk of radius $r_h$ centered at $Y\in\Phi_b$. All the PBs inside this region contribute to the  received power $P_{BN}$ at the BN. We further assume that the power received from the PBs which are outside the harvesting region is negligible.  The number of PBs $\bar{N}_p$ inside $z(Y,r_h)$ is a Poisson random variable with mean ${N}_p=\pi r_h^2\lambda_p$. In order to avoid the expectation with respect to the probability density function (PDF) of $\bar{N}_p$ in the analysis, we use the average number of PBs $N_p$ in the rest of the paper. 
We adjust $r_h$ such that  $N_p$ is an integer number. The received power  at $\text{BN}_0$ can be written as
\begin{equation}
P_{{BN}_0} = \left|\sum_{i=1}^{{N}_p}\sqrt{P_C}h_{0,i}^f x_{0,i}^{-\alpha^f/2}\right|^2,
\label{px_pb}
\end{equation} 
where $P_C$ is the transmit power of a PB, $h_{0,i}^f$ represents the forward channel gain between $\text{BN}_0$ and  its $i$th nearest PB  and $h_{0,i}^f\sim \mathrm{CN}\left(0,1\right)$, $\alpha^f$ represents the path loss exponent of the forward channel i.e., between the PBs and $\text{BN}_0$,  and $x_{0,i}$ is the distance between  $\text{BN}_0$ and its $i$th nearest PB. The received signal by $\text{BN}_0$ is modulated and reflected to the receiver. The portion of the received power  reflected back by $\text{BN}_0$ is modeled with a reflection coefficient $\beta\in\left[0,1\right]$, which we call as backscattering efficiency. $\text{BN}_0$ reflects $\beta P_{{BN}_0}$ amount of power while the remaining power $\left(1-\beta\right) P_{{BN}_0}$ is harvested (as shown in Fig. \ref{BN_node}) \cite{boyer, bletsas}. 
The harvested power is used by $\text{BN}_0$ to run its  integrated circuit (IC). 

The joint PDF of the distances from $\text{BN}_0$ to its nearest $N_p$ PBs  will be used throughout the paper, therefore, we present it here 
\begin{equation}
f_x(x_{0,1},\cdots,x_{0,N_p}) =  \left(2\pi\lambda_p\right)^{N_p} x_{0,1},\dots, x_{0,N_p} \mathrm{e}^{-\pi\lambda_p x_{0, N_p}^2},
\label{joint_pdf_dist}
\end{equation}
where $x_{0,1} < x_{0,2}\cdots < x_{0,N_p}$ \cite{distance_distribution}. The joint PDF of the distances between the non-typical BNs to their nearest $N_p$ PBs is the same as \eqref{joint_pdf_dist}; however, to make the distinction clear, we represent it by $f_x(x_{1},\cdots,x_{N_p})$ throughout the paper.
The SINR $\gamma_{R_0}$ at a receiver $\text{R}_0$, when the $\text{BN}_0$ transmits under the above system model and in the presence of interference from other active BNs can be written as
\begin{equation}
\gamma_{R_0}' = \frac{\beta \textbf{h}_{0,0}^b d_{0,0}^{-\alpha^b}\left|\sum_{i=1}^{N_p}\sqrt{P_C}h_{0,i}^f x_{0,i}^{-\alpha^f/2}\right|^2}{\beta \sum_{j\in\Phi_b\backslash BN_0}\textbf{h}_{0,j}^b d_{0,j}^{-\alpha^b}\left|\sum_{k=1}^{N_p}\sqrt{P_C}h_{j,k}^f x_{j,k}^{-\alpha^f/2}\right|^2+N_0},
\label{sinr_first}
\end{equation}
where $\textbf{h}_{i,j}^b =\left|h_{i,j}^b\right|^2$ is the backward channel gain between $i$th receiver and $j$th BN and $h_{i,j}^b\sim \mathrm{CN}\left(0,1\right)$;  $h_{m,n}^f$ is the forward channel between $m$th BN and $n$th PB and $h_{m,n}^f \sim \mathrm{CN}\left(0,1\right)$; $d_{i,j}$ is the distance between $i$th receiver and $j$th BN; $x_{m,n}$ is the distance between $m$th BN and $n$th PB; $\alpha^b$ is the path loss exponent of the backward channel i.e., between the BN and the receiver; $N_0$ is the variance of the additive white Gaussian noise (AWGN). In order to avoid the singularity in the forward link (infinite received power at a BN), we use the bounded path loss model i.e., the path loss is $\min\left\{1, x_{i,j}^{-\alpha^f}\right\}$  \cite{martin_book}. We assume that the distance between the receiver and the typical BS $d_{0,0}$ is fixed. We consider Rayleigh fading and assume that all channel gains are independent and identically distributed (i.i.d).
%
%
\subsection{Performance Metrics}
We consider two metrics to study the performance of this wireless powered backscatter communication network. The first metric is the coverage probability, $\mathbb{P}_s$, which is the success probability that the receiver $\text{R}_0$ can successfully decode  the signal of  $\text{BN}_0$. Given a threshold $\Theta_m$, we defined the coverage probability as
\begin{equation}
\mathbb{P}_s = \mathbb{P}\left[\gamma_{R_0}'\geq\Theta_{m}\right],
\label{success_prob}
\end{equation} 
which gives us the percentage of BNs having successful transmission in the network.  
Our second metric is the transmission capacity, $\mathbb{C}$, defined as
\begin{equation}
\mathbb{C} = \lambda_b \mathbb{P}_s,
\label{trans_capacity}
\end{equation}
which is the number of successful transmissions between the BNs and their corresponding receivers per unit area of the network. 
\section{Coverage Probability Analysis}
In this section, we find the coverage probability, $\mathbb{P}_s$, and the capacity, $\mathbb{C}$, of the network. For the coverage probability, we consider two cases. In the first case, we consider that a BN harvests energy from the nearest $N_p$ PBs as mentioned in the previous section, while for the second case, we assume that the BNs harvest energy from all the PBs in the network. The analysis depends on both the PPPs of PBs $\Phi_p$ and of the BNs $\Phi_b$, and both the forward channel $\textbf{h}_{i,j}^f$ and the backward channel $\textbf{h}_{k,i}^b$. Whereas, in conventional wireless network, we do not have this extra tier of PBs and both the BS and the MT have reliable power source. The downlink analysis usually depends on the PPP of the BSs and the channel between the MT and the BS, while the uplink analysis depends on the PPP of MTs and the channel. Thus having the additional tier of PBs and the backscatter communications by the BNs make the analysis very challenging. Therefore, in order to make the analysis tractable, we make some assumption as discussed in the sequel. The second case is more complicated than the first one, therefore, we find an approximation for the coverage probability. The details of the derivation are given in the following subsections.

\subsection{Coverage Probability when a BN harvests energy from $N_p$ nearby PBs}
The transmission of a typical BN is successful when the SINR $\gamma_{R_0}'$ exceeds the threshold $\Theta_m$. It is very difficult to find  the coverage probability of the network due to double fading (forward and backward channel) and its dependence on $\Phi_b$ and $N_p$ nearest PBs. Therefore, in order to make the analysis tractable, we use a slightly modified definition of the SINR. The received power at the typical BN $\left|\sum_{i=1}^{{N}_p}\sqrt{P_C}h_{0,i}^f x_{0,i}^{-\alpha^f/2}\right|^2$ is exponentially distributed with mean $P_C\sum_{i=1}^{{N}_p} x_{0,i}^{-\alpha^f}$. We replace the harvested power with its mean value in \eqref{sinr_first}. The accuracy of this assumption is verified through simulations in the next section. It is important to mention that by doing so, we take the expectation only with respect to the forward channel and not with the location of the $N_p$ PBs. Similar approaches have been used in the literature to reduce the dimensionality of the problem when the analysis is very complicated and leads to intractable results \cite{itlinq,analytical_D2D,joint_harpeet,jeff_tit,shotgun_cellular}. More specifically in \cite{itlinq}, and \cite{analytical_D2D}, the authors considered an averaging circle around the receiver of interest and  model the interferers inside this region while for the interference emanating from outside the region they used the expected value of it. Similarly, in \cite{joint_harpeet}, the authors only considered the effect of two nearest interferers and for the rest of the interference, its  average value is considered. 
We write the modified SINR $\gamma_{R_0}$ as 
\begin{equation}
\gamma_{R_0} = \frac{ \textbf{h}_{0,0}^b d_{0,0}^{-\alpha^b}\sum_{i=1}^{N_p} x_{0,i}^{-\alpha^f}}{\sum_{j\in\Phi_b\backslash BN_0}\textbf{h}_{0,j}^b d_{0,j}^{-\alpha^b}\sum_{k=1}^{N_p} x_{j,k}^{-\alpha^f}+\frac{N_0}{\beta P_C}}.
\label{sinr_2nd}
\end{equation}

Before deriving our main result for the coverage probability,  we first find the average received power by a random BN in the following lemma.
\begin{lem}
The average received power at  $\text{BN}_0$ when there are $N_p$ PBs in the harvesting zone is given by
\begin{equation}
\bar{P}_{{BN}_0}= P_C\left(2\pi\lambda_p\right)^{N_p}  \idotsint\displaylimits_{0<x_{0,1}<x_{0,2}<\cdots<x_{0,N_p}<\infty}x_{0,1},\dots, x_{0,N_p} \left(\sum_{i=1}^{{N}_p} x_{0,i}^{-\alpha^f}\right) \mathrm{e}^{-\pi\lambda_p x_{0,N_p}^2} \mathrm{d}x_{0,1},\cdots,\mathrm{d}x_{0,N_p}.
\label{average_power1}
\end{equation}
\label{avg_power_1}
\end{lem}
\begin{proof}
The average received power at a BN can written as
\begin{equation}
\bar{P}_{{BN}_0}=\mathbb{E}\left[P_{{BN}_0}\right] = \mathbb{E}_{h_i^f, x_{0,1},\cdots,x_{0,N_p}}\left[\left|\sum_{i=1}^{{N}_p}\sqrt{P_C}h_{0,i}^f x_{0,i}^{-\alpha^f/2}\right|^2\right] \overset{a}=\mathbb{E}_{h_{0,i}^f, x_{0,1},\cdots,x_{0,N_p}}\left[P_C\sum_{i=1}^{{N}_p} x_{0,i}^{-\alpha^f}\right], 
\end{equation}
where $\left(a\right)$ follows because $h_{0,i}^f$ is i.i.d distributed  $\mathrm{CN}\left(0,1\right)$ and after taking the expectation with respect to $x_{0,1},\cdots,x_{0,N_p}$ we obtained the final result. 
\end{proof}
It can be noticed from \eqref{average_power1} that $\bar{P}_{{BN}_0}$ increases with an increase in $P_C$, $N_p$ and $\lambda_p$.  It is important to remind that on average, a BN harvests $\left(1-\beta\right)\bar{P}_{{BN}_0}$ while $\beta \bar{P}_{{BN}_0}$ is reflected back through backscatter modulation. It should be noted also that we do not use the dimensionality reduction in the proof of Lemma \ref{avg_power_1} and the expression in \eqref{average_power1}  is exact.

It should be noted that our work significantly differs from \cite{coordinated_joint_trans, spatiotemporal_cooperation}. Both \cite{coordinated_joint_trans, spatiotemporal_cooperation} consider a cellular network while we consider an ad-hoc setup. In our work, each BN harvests energy from the PBs in its harvesting zone and then backscatters the signal to the receiver, while in \cite{coordinated_joint_trans, spatiotemporal_cooperation} multiple BSs cooperatively transmit to a receiver. Moreover, the authors \cite{coordinated_joint_trans, spatiotemporal_cooperation} first map the PPPs to a single one-dimensional PPP and then carry out the analysis, while we directly apply the probability generating functional (PGFL) of the 2-D PPP in the proof of Theorem \ref{theorem1} and Corollary \ref{cor_1}.  We present the SINR coverage probability in the following theorem.
\begin{thm}
The SINR coverage probability $\mathbb{P}_s$ when the BNs harvests from the nearest $N_p$ PBs is given by
\begin{equation}
\mathbb{P}_s = \idotsint\displaylimits_{0<x_{0,1}<x_{0,2}<\cdots<x_{0,N_p}<\infty} \exp(-\frac{s N_0}{\beta P_C})\mathcal{L}_I(s)f_x(x_{0,1},\cdots,x_{0,N_p})\mathrm{d}x_{0,1},\cdots,\mathrm{d}x_{0,N_p}
\label{thm1}
\end{equation}
where $s=\frac{\Theta_{m}d_{0,0}^{\alpha^f}}{\sum_{i=1}^{N_p}x_{0,i}^{-\alpha^f}}$, $\mathcal{L}_I(s)$ is the Laplace transform of the interference and is given by 
\begin{equation}
\mathcal{L}_I(s) =\exp\left(-\frac{2\pi^2 \lambda_b s^{2/\alpha^b}}{\alpha^b \sin\left(\frac{2 \pi}{\alpha^b}\right)} \mathbb{E}_{x_1,\cdots,x_{N_p}}\left(\sum_{k=1}^{N_p} x_k^{-\alpha^f}\right)^{2/\alpha^b}\right),
\label{lap_thm1_final}
\end{equation}
and $f_x(x_{0,1},\cdots,x_{0,N_p})$ is the joint PDF of the distances to the $N_p$ nearest PBs and given is in \eqref{joint_pdf_dist}. 
\label{theorem1}
\end{thm}
\begin{proof}
The proof is provided in Appendix A. 
\end{proof}
\begin{remark}
(Effect of $\beta$ and $P_C$ on the coverage probability) It can be observed from \eqref{thm1} that the coverage probability $\mathbb{P}_s$ increases by increasing the backscattering efficiency $\beta$ or by increasing the transmit power $P_C$ of the PBs.
\label{remark1}
\end{remark}
We observe that, as the number of PBs $N_p$ in the harvesting region increases, the numerical computation of the coverage probability $\mathbb{P}_s$ becomes very tedious. This is due to the integration over the joint PDF of the distances to the $N_p$ nearest PBs both in the expression for the  $\mathcal{L}_I(s)$ and in the final expression of \eqref{thm1}. Nonetheless, the coverage probability $\mathbb{P}_s$ can be simplified for the following plausible special cases. 
\begin{cor}
When the BN harvests from the single nearest PB $\left(N_p=1\right)$ and $\alpha^f=\alpha^b=\alpha$, then the coverage probability simplifies to
\begin{equation}
\mathbb{P}_s =2\pi \lambda_P \int_0^\infty x_{0,1} \exp\left(-\left[\frac{s N_0}{\beta P_C}+\pi \lambda_p x_{0,1}^{2} + \frac{2\pi^2 \lambda_b s^{2/\alpha}}{\alpha \sin\left(\frac{2 \pi}{\alpha}\right)}\left[1-\mathrm{e}^{-\pi \lambda_p}-\pi \lambda_p \mathrm{Ei}\left(-\pi \lambda_p\right)\right]\right]\right)\mathrm{d}x_{0,1},
\label{cor1}
\end{equation}
where $s=\Theta_{m}d_{0,0}^{\alpha}x_{0,1}^{\alpha}$ and $\mathrm{Ei}\left(z\right)=\int_z^{\infty}\mathrm{e}^{-t}t^{-1}\mathrm{d}t$ is the exponential integral function \cite{book_integral}. The coverage probability can be computed by evaluating just a single integral.
\label{cor_1}
\end{cor}
\begin{proof}
The proof follows similar steps as in the proof of Theorem 1. However, to be self-contained, we provide the main steps here. 
By using the marginal PDF of $x_{0,1}$, we write $\mathbb{P}_s$ as
\begin{equation}
\mathbb{P}_s = 2\pi \lambda_P \int_0^\infty x_{0,1} \mathbb{P}\left[\gamma_{R_0}\geq\Theta_{m} | x_{0,1} \right] \exp\left(-\pi\lambda_p x_{0,1}^2\right)\mathrm{d}x_{0,1},
\label{cor1_proof_eq1}
\end{equation}
where, by utilizing the definition of $\gamma_{R_0}$, the probability $\mathbb{P}\left[\gamma_{R_0}\geq\Theta_{m} | x_{0,1} \right]$ can be written as
\begin{multline}
\mathbb{P}\left[\frac{ \textbf{h}_{0,0}^b d_{0,0}^{-\alpha} x_{0,1}^{-\alpha}}{\sum_{j\in\Phi_b\backslash BN_0}\textbf{h}_{0,j}^b d_{0,j}^{-\alpha} x_{j,1}^{-\alpha}+\frac{N_0}{\beta P_C}} \geq \Theta_{m} \right] = \mathbb{P}\left[\textbf{h}_{0,0}^b \geq s \left(I+\frac{N_0}{\beta P_C}\right)\right]= \exp(-\frac{s N_0}{\beta P_C})\mathcal{L}_I(s),
\label{cor1_proof_eq2}
\end{multline}
with $s=\Theta_{m}d_{0,0}^{\alpha}x_{0,1}^{\alpha}$ and $I = \sum_{j\in\Phi_b\backslash BN_0}\textbf{h}_{0,j}^b d_{0,j}^{-\alpha} x_{j,1}^{-\alpha}$. Now, to find $\mathcal{L}_I(s)$, we apply the same steps as we did  to obtain \eqref{lap_thm1} in the proof of Theorem \ref{theorem1} and we get 
\begin{multline}
\mathcal{L}_I(s) = \exp\left(-\frac{2\pi^2 \lambda_b s^{2/\alpha}}{\alpha \sin\left(\frac{2 \pi}{\alpha}\right)} \mathbb{E}_{x_1}\left( x_1^{-2}\right)\right) \\ = \exp\left(-\frac{2\pi^2 \lambda_b s^{2/\alpha}}{\alpha \sin\left(\frac{2 \pi}{\alpha}\right)} \left(2\pi \lambda_p \int_0^1  x_1 \mathrm{e}^{-\pi\lambda_p x_1^2} \mathrm{d}x_{1}+ 2\pi \lambda_p \int_1^\infty  x_1^{-1} \mathrm{e}^{-\pi\lambda_p x_1^2} \mathrm{d}x_{1} \right)\right),
\label{cor1_proof_eq3}
\end{multline}
where the last expression is obtained by utilizing the bounded path loss model for $\mathbb{E}_{x_1}\left( x_1^{-2}\right)$. The first integral  $2\pi \lambda_p \int_0^1  x_1 \mathrm{e}^{-\pi\lambda_p x_1^2} \mathrm{d}x_{1}= 1-\mathrm{e}^{-\pi \lambda_p}$, and the second integral is an exponential integral and  $2\pi \lambda_p \int_1^\infty  x_1^{-1} \mathrm{e}^{-\pi\lambda_p x_1^2} \mathrm{d}x_{1} =-\pi \lambda_p \mathrm{Ei}\left(-\pi \lambda_p\right)$, where  $\mathrm{Ei}\left(z\right)$ is the exponential integral function and  $\mathrm{Ei}\left(z\right)=\int_z^{\infty}\mathrm{e}^{-t}t^{-1}\mathrm{d}t$ \cite{book_integral}. Thus evaluating both integrals in the above expression, plugging  \eqref{cor1_proof_eq3} in \eqref{cor1_proof_eq2} and then plugging back \eqref{cor1_proof_eq2} in \eqref{cor1_proof_eq1}  and doing some manipulations, we obtain \eqref{cor1}.
\end{proof}
\begin{remark}
(Effect of $\lambda_p$ on the coverage probability) The density of the PBs $\lambda_p$ appears both inside the exponential and outside the exponential in \eqref{cor1}. We expect that when $\Theta_m$ is small then the $\lambda_p$ outside the exponential dominates and as a result the coverage probability increases with an increase in $\lambda_p$. However, when $\Theta_m$ is large enough then the exponential-term dominates due to which the coverage probability decreases with an increase in $\lambda_p$.
\label{remark2}
\end{remark}
\begin{remark}
(Effect of $\lambda_b$ on the coverage probability) The density of the BNs $\lambda_b$ is only inside the exponential in \eqref{cor1}, which suggests that an increase in $\lambda_b$ always decreases the coverage probability. 
\label{remark3}
\end{remark}
\begin{remark}
(Effect of $\lambda_b$ on the capacity of the network) The capacity $\mathbb{C} = \lambda_b \mathbb{P}_s$, which shows that the density of BNs $\lambda_b$ is both inside the exponential term (in the expression for $\mathbb{P}_s$) and outside the exponential term. This means that the $\lambda_b$ outside and inside exponential affects the capacity $\mathbb{C}$ in a completely different way. We expect that for small $\Theta_m$ it will increase $\mathbb{C}$, whereas for large $\Theta_m$ it will decrease $\mathbb{C}$.
\label{remark4}
\end{remark}
\begin{cor}
When $N_p=1, N_0=0$ and $\alpha^f=\alpha^b=\alpha$ then \eqref{cor1} can be further simplified to 
\begin{equation}
\mathbb{P}_s =2\pi \lambda_P \int_0^\infty x_{0,1} \exp\left(-\left[\pi \lambda_p x_{0,1}^{2} + \frac{2\pi^2 \lambda_b s^{2/\alpha}}{\alpha \sin\left(\frac{2 \pi}{\alpha}\right)}\left[1-\mathrm{e}^{-\pi \lambda_p}-\pi \lambda_p \mathrm{Ei}\left(-\pi \lambda_p\right)\right]\right]\right)\mathrm{d}x_{0,1},
\label{cor2}
\end{equation}
where $s=\Theta_{m}d_{0,0}^{\alpha}x_{0,1}^{\alpha}$. The coverage probability becomes independent of the transmit power $P_C$ of the PB  and the backscattering efficiency $\beta$. 
\label{cor_2}
\end{cor}
\begin{cor}
When $N_0=0$, but the BNs harvest energy from multiple nearby PBs then the coverage probability simplifies to 
\begin{multline}
\mathbb{P}_s = \idotsint\displaylimits_{0<x_{0,1}<x_{0,2}<\cdots<x_{0,N_p}<\infty}\exp\left(-\frac{2\pi^2 \lambda_b s^{2/\alpha^b}}{\alpha^b \sin\left(\frac{2 \pi}{\alpha^b}\right)} \mathbb{E}_{x_1,\cdots,x_{N_p}}\left(\sum_{k=1}^{N_p} x_k^{-\alpha^f}\right)^{2/\alpha^b}\right) \times \\ f_x(x_{0,1},\cdots,x_{0,N_p})\mathrm{d}x_{0,1}\cdots\mathrm{d}x_{0,N_p},
\label{cor3}
\end{multline}
where $s=\frac{\Theta_{m}d_{0,0}^{\alpha^b}}{\sum_{i=1}^{N_p}x_{0,i}^{-\alpha^f}}$ and $\mathbb{P}_s$ is independent of the transmit power $P_C$ of PBs and the backscattering efficiency $\beta$.
\label{cor_3}
\end{cor}
\begin{cor}
The coverage probability can be further simplified to 
\begin{equation}
\mathbb{P}_s = \exp\left(-s \left[\frac{ N_0}{\beta \bar{P}{{BN}_0}} + \frac{2 \pi^2 \lambda_b s^{2/\alpha^b-1}}{\alpha^b \sin\left(\frac{2 \pi}{\alpha^b}\right)}\right]\right),
\label{cor4}
\end{equation}
if the BNs reflect the average power $\beta\bar{P}_{{BN}_0}$  instead of the instantaneous power $\beta P_{{BN}_0}$.  The coverage probability is in closed form. The network behaves like a conventional ad-hoc network with BNs having the transmit power $\beta \bar{P}_{{BN}_0}$  and $s=\Theta_m d_{0,0}^{\alpha^b}$.
\label{cor_4}
\end{cor}
\begin{proof}
We first put the average harvested power  $\bar{P}_{{BN}_0}$ instead of the instantaneous $\beta P_{{BN}_0}$  in \eqref{sinr_2nd} and the rest of the proof follows similar steps as in Theorem $1$.
\end{proof}
\begin{cor}
When $N_0=0$ and the BNs reflect the average power $\beta \bar{P}_{{BN}_0}$, then the coverage probability further simplifies to 
\begin{equation}
\mathbb{P}_s = \exp\left(-\frac{2 \pi^2 \lambda_b \Theta_m^{2/\alpha^b} d_{0,0}^2}{\alpha^b \sin\left(\frac{2 \pi}{\alpha^b}\right)}\right),
\label{cor5}
\end{equation}
which depends only on the path loss exponent $\alpha^b$, density of the BN $\lambda_b$, and the distance $d_{0,0}$ between the BN and its receiver. The network acts like  a conventional ad-hoc network and the coverage is independent of the transmit power of the BNs.
\label{cor_5}
\end{cor}
\subsection{Coverage Probability when a BN harvests energy from all PBs}
In contrast to the previous subsection, here we assume that a BN harvests energy from all the PBs in the network. We write the received power at the typical BN as
\begin{equation}
P_{{BN}_0} = \left|\sum_{i\in\Phi_p}\sqrt{P_C} h_{0,i}^f x_{0,i}^{-\alpha^f/2}\right|^2,
\label{px_pb_2}
\end{equation} 
where the summation is over the $\Phi_p$. Similarly, we write the SINR $\gamma_{R_0}'$ as
\begin{equation}
\gamma_{R_0}' = \frac{\beta \textbf{h}_{0,0}^b d_{0,0}^{-\alpha^b}\left|\sum_{i\in\Phi_p}\sqrt{P_C} h_{0,i}^f x_{0,i}^{-\alpha^f/2}\right|^2}{\beta \sum_{j\in\Phi_b\backslash BN_0}\textbf{h}_{0,j}^b d_{0,j}^{-\alpha^b}\left|\sum_{k\in\Phi_p}\sqrt{P_C}h_{j,k}^f x_{j,k}^{-\alpha^f/2}\right|^2+N_0}.
\label{sinr_second}
\end{equation}
It is more challenging than the previous section to find the coverage probability based on the above SINR expression because it depends on the PPPs of both PBs and BNs, and both the forward and backward channel. Therefore, in order to make the analysis tractable,  we make some modifications to the above expression and consider SIR instead of SINR. We write the modified expression for SIR as
\begin{equation}
\gamma_{R_0} = \frac{ \textbf{h}_{0,0}^b d_{0,0}^{-\alpha^b}\sum_{i\in\Phi_p} x_{0,i}^{-\alpha^f}}{\sum_{j\in\Phi_b\backslash BN_0}\textbf{h}_{0,j}^b d_{0,j}^{-\alpha^b}\sum_{k\in\Phi_p}x_{j,k}^{-\alpha^f}},
\label{sir_second}
\end{equation} 
where we use the average of $\left|\sum_{i\in\Phi_p}\sqrt{P_C} h_{0,i}^f x_{0,i}^{-\alpha^f/2}\right|^2$, which is $P_C\sum_{i\in\Phi_p} x_{0,i}^{-\alpha^f}$ \cite{martin_book}. It is important to mention that in \eqref{sir_second}, we only take the average with respect to the forward channel and not the $\Phi_p$. The coverage probability still depends on the $\Phi_p, \Phi_b$ and $\textbf{h}_{j,k}^b$. Similar approaches have been used in \cite{joint_harpeet,itlinq,analytical_D2D,jeff_tit,shotgun_cellular}  to reduce the complexity of the problem and to get tractable results. 

Before presenting the main result of this section, we first find the average power at a BN in the following lemma.
\begin{lem}
The average power at a BN when it harvests from all PBs in the network is given by
\begin{equation}
\bar{P}_{{BN}_0} = \frac{P_C \pi \lambda_p   \alpha^f}{\alpha^f-2}.
\label{avg_power2}
\end{equation}
\label{avg_power_2}
\end{lem}
It is important to mention that a power $\left(1-\beta\right)\bar{P}_{{BN}_0}$  is harvested while the remaining power $\beta \bar{P}_{{BN}_0}$ is reflected through backscatter modulation. Similar to Lemma \ref{avg_power_1}, we do not use the dimensionality reduction in the proof of Lemma  \ref{avg_power_2} and the expression in \eqref{avg_power2} is exact. 
\begin{proof}
The average harvested power $\bar{P}_{BN}$ in this setup can be written as
\begin{equation}
\bar{P}_{{BN}_0} = \mathbb{E}\left[P_{{BN}_0}\right] = \mathbb{E}_{h_{0,i}^f, \Phi_p}\left[\left|\sum_{i\in\Phi_p}\sqrt{P_C} h_{0,i}^f x_{0,i}^{-\alpha^f/2}\right|^2\right]\overset{a}=\mathbb{E}_{\Phi_p}\left[P_C\sum_{i\in\Phi_p} x_{0,i}^{-\alpha^f}\right],
\end{equation}
where $a$ follows because $h_{0,i}^f$ is i.i.d distributed $\mathrm{CN}\left(0,1\right)$ and the final result follows from Campbell's theorem \cite{martin_book}.
\end{proof}
We know that the SIR coverage probability depends on both the $\Phi_p$, $\Phi_b$, and its exact expression cannot be obtained. Therefore, we utilize Jensen's inequality and find the approximated expression for the coverage probability in the following theorem.
\begin{thm}
The SIR coverage probability when the BNs harvest energy from all the PBs is given by
\begin{equation}
\mathbb{P}_s \approx \exp\left(-\frac{2 \pi^2 \lambda_b \Theta_m^{2/\alpha^b} d_{0,0}^2}{\alpha^b \sin\left(\frac{2 \pi}{\alpha^b}\right)}\right).
\label{eq:thm2}
\end{equation}
\label{thm2}
\end{thm}
\begin{proof}
See Appendix B.
\end{proof}
The above expression is the same as of Corollary \ref{cor_5}, which makes sense because \eqref{cor5} is the SIR coverage probability when the BNs use the average transmit power and we know that SIR coverage is independent of the transmit power of the BNs.
\section{Simulation Results and Discussion}
\begin{figure*} 
\centering
			\begin{subfigure}{0.5\textwidth}
			\centering
			\includegraphics[width=1\linewidth]{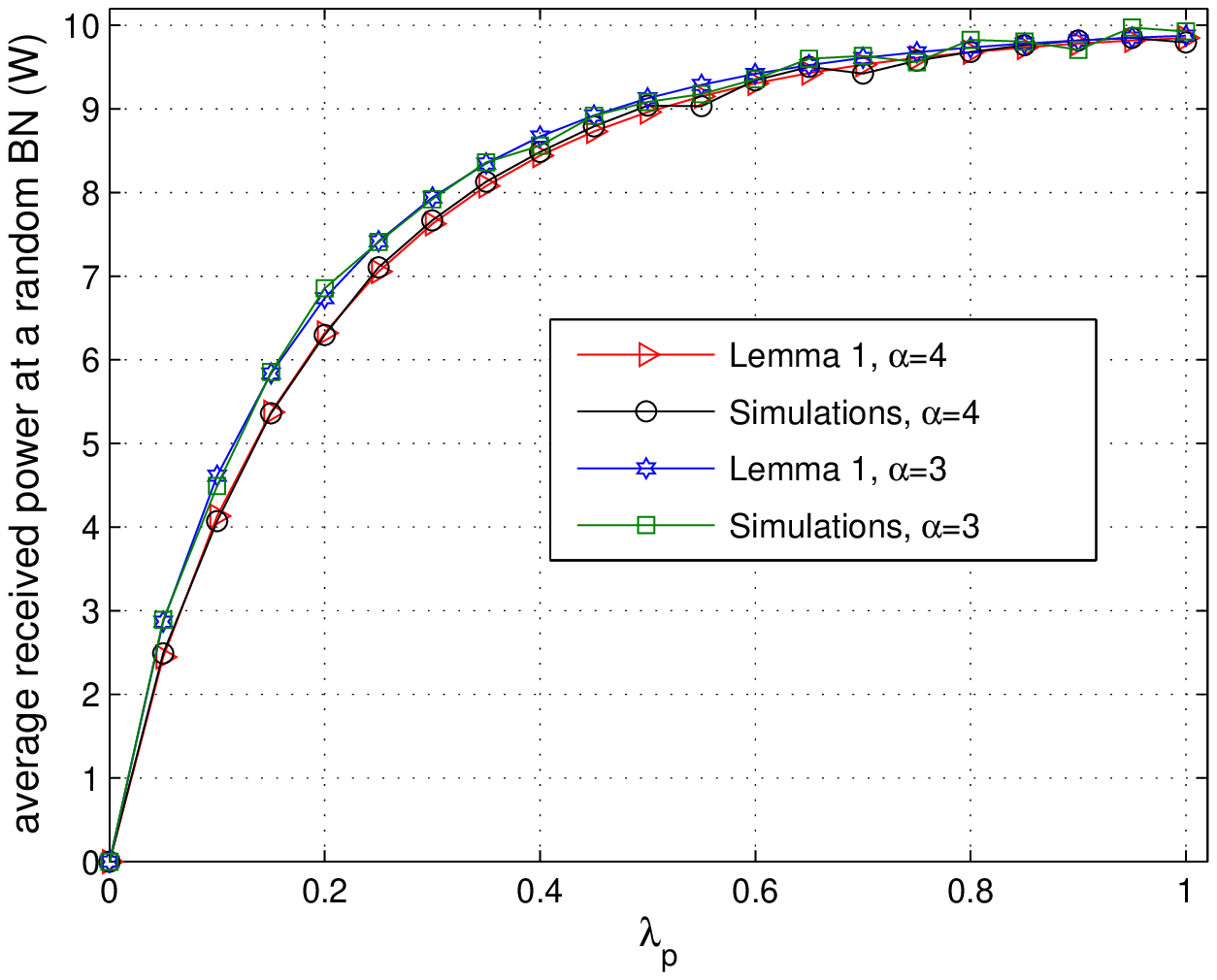}
			\caption{Lemma 1 (harvesting energy from the nearest PB)}
			\label{sub:fig_lemma1}
		\end{subfigure}
	\begin{subfigure}{0.5\textwidth}
		\centering
		\includegraphics[width=1\linewidth]{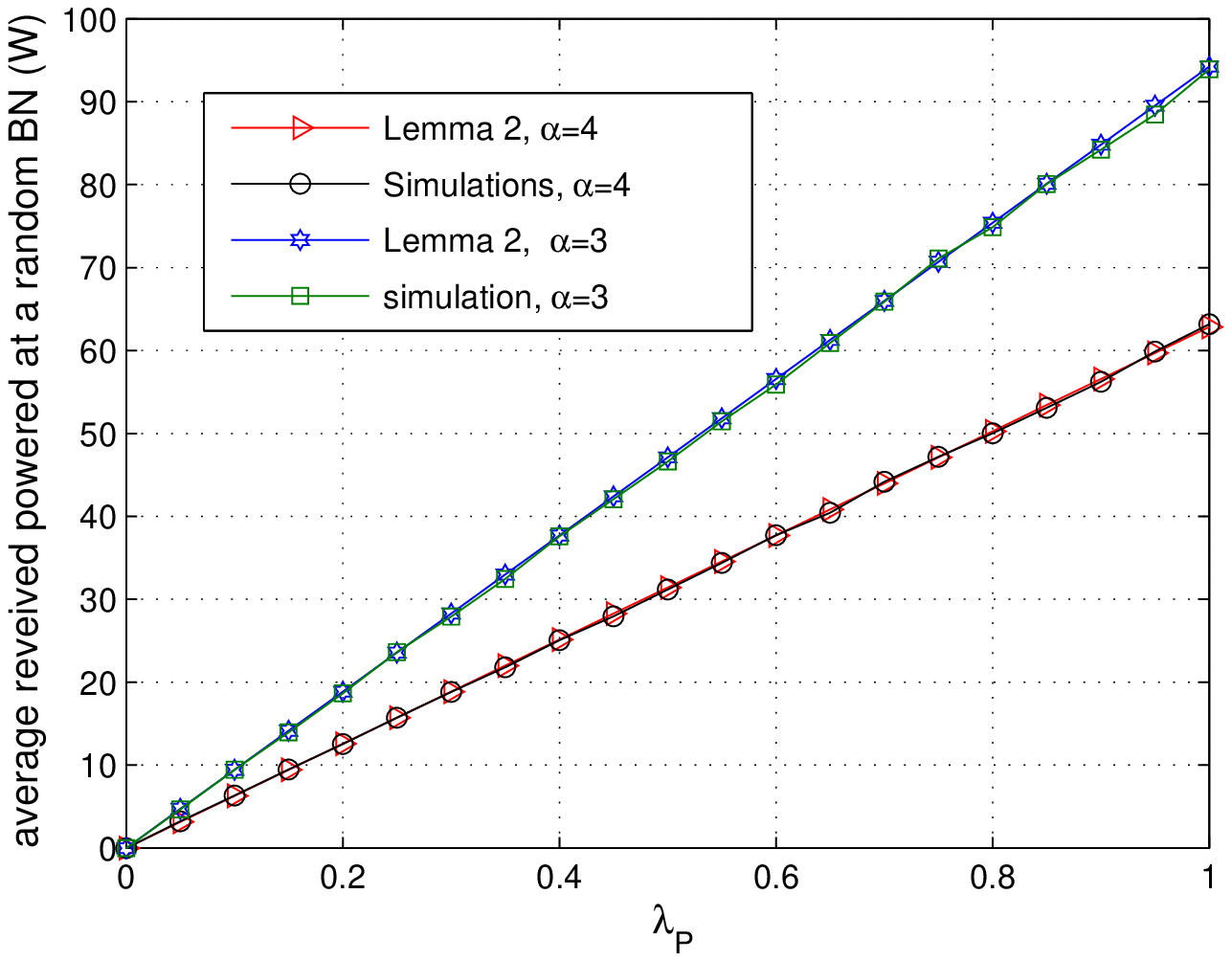}
		 \caption{Lemma 2 (harvesting energy from the whole network)}
		 \label{sub:fig_lemma2}
	\end{subfigure}\hfill
\caption{Simulations vs analysis (Lemma 1 and Lemma 2)}
		\label{fig:sim_vs_anal_lemma1_and_2}
\end{figure*}
In this section, we present numerical results to validate the accuracy of the analysis and the tradeoffs that emerge in this wireless powered backscatter network. The system model of Section II are carefully reproduced in the Monte Carlo simulations. We fix $P_C=40$dBm, $\beta=0.5$, $d_{0,0}=1$, $N=-40$dB, and $\alpha^f=\alpha^b=\alpha$. Both $\lambda_p$ and $\lambda_b$ are in per square meters $\left(/m^2\right)$ and we consider a unit area of $100m^2$ when  computing the capacity $\mathbb{C}$ of the network. If we do not mention $\alpha$ and $N_p$ then we keep $\alpha=4$ and $N_p=1$, otherwise, we state it. 

In Fig. \ref{fig:sim_vs_anal_lemma1_and_2}, we plot the average received power at a BN $\bar{P}_{BN}$ against the density of the PBs $\lambda_p$. It can be seen that the numerical results obtained through Lemma 1 and 2 exactly match with the simulations results. It can be observed that $\bar{P}_{BN}$ is greater for small path loss exponent i.e., $\alpha=3$. Furthermore, it can be noticed that $\bar{P}_{BN}$ increases with an increase in $\lambda_p$ for both Fig. \ref{sub:fig_lemma1} and Fig. \ref{sub:fig_lemma2}. However, for Lemma 1 in Fig. \ref{sub:fig_lemma1}, it saturates to $10$watts ($40$dBm) for $\lambda_p=1$, which is the transmit power $P_C$ of a PB. This implies that when $\lambda_p=1$, $N_p=1$ and the path loss model is $\min\left\{1, x^{-\alpha}\right\}$,  the received power is on average equal to the transmit power of a PB. Whereas, in Lemma 2 the BNs harvest power from all the PBs, therefore, in Fig. \ref{sub:fig_lemma2} the $\bar{P}_{BN}$ does not saturate and keep increasing with $\lambda_p$.
\begin{figure}
	\centering
		\includegraphics[scale=0.6]{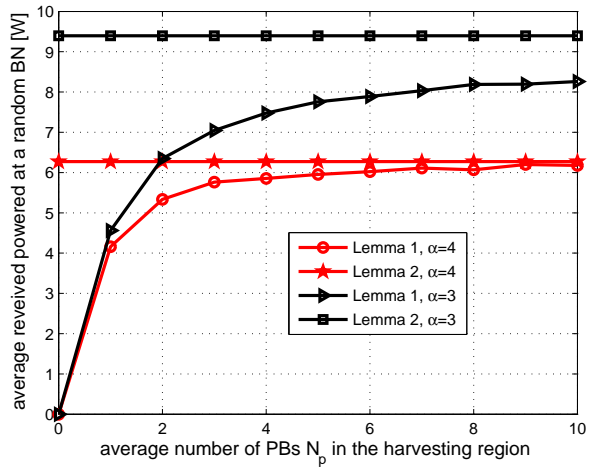}
		\caption{Impact of the number of PBs $N_P$ in the harvesting zone on the $\bar{P}_{BN}$, ($\lambda_p=0.1$)}
		\label{fig:lemma1_2}
\end{figure}
\begin{figure*} 
\centering
			\begin{subfigure}{0.5\textwidth}
			\centering
			\includegraphics[width=1\linewidth]{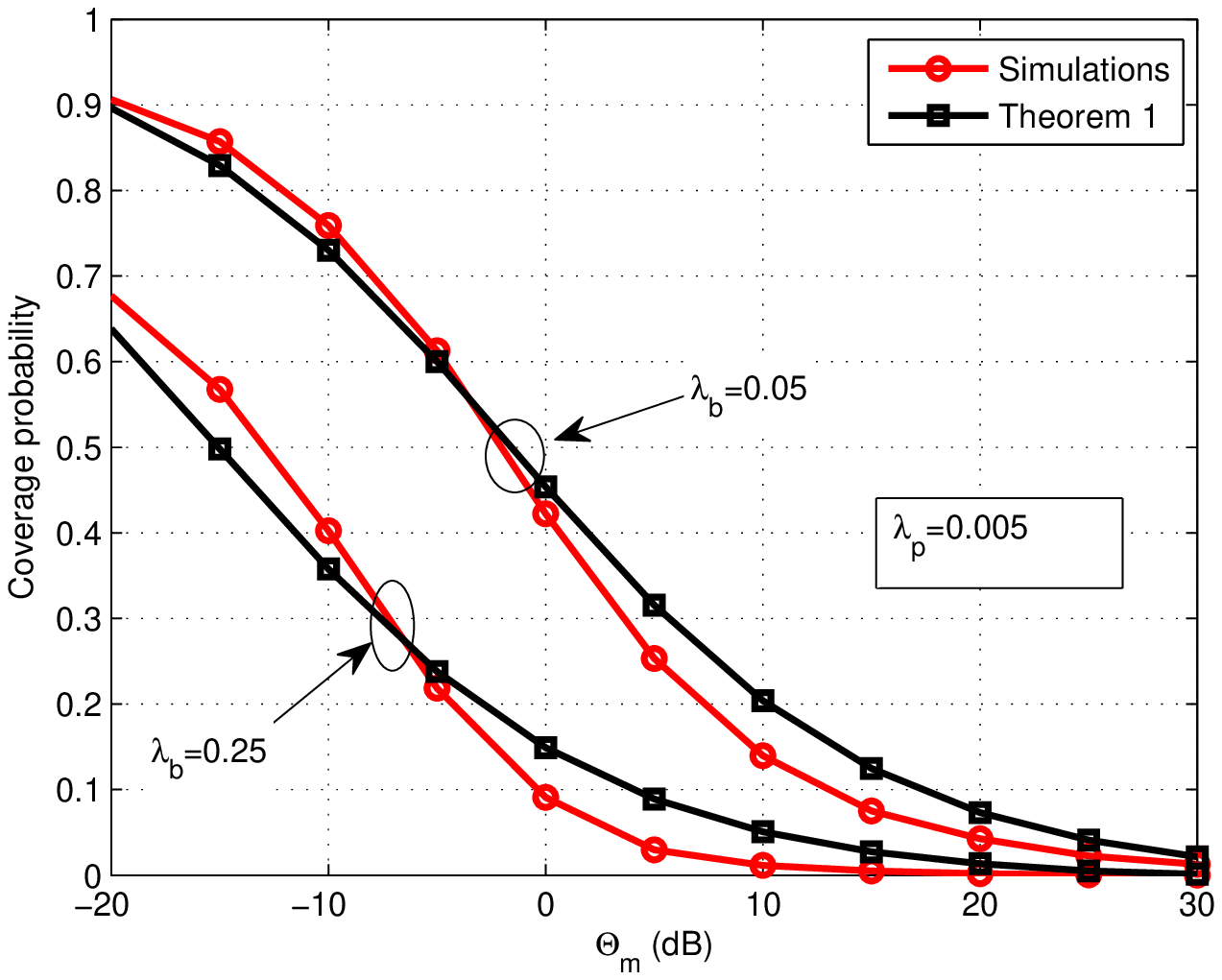}
			\caption{Theorem 1}
			\label{sub:fig_thm1}
		\end{subfigure}
	\begin{subfigure}{0.5\textwidth}
		\centering
		\includegraphics[width=1\linewidth]{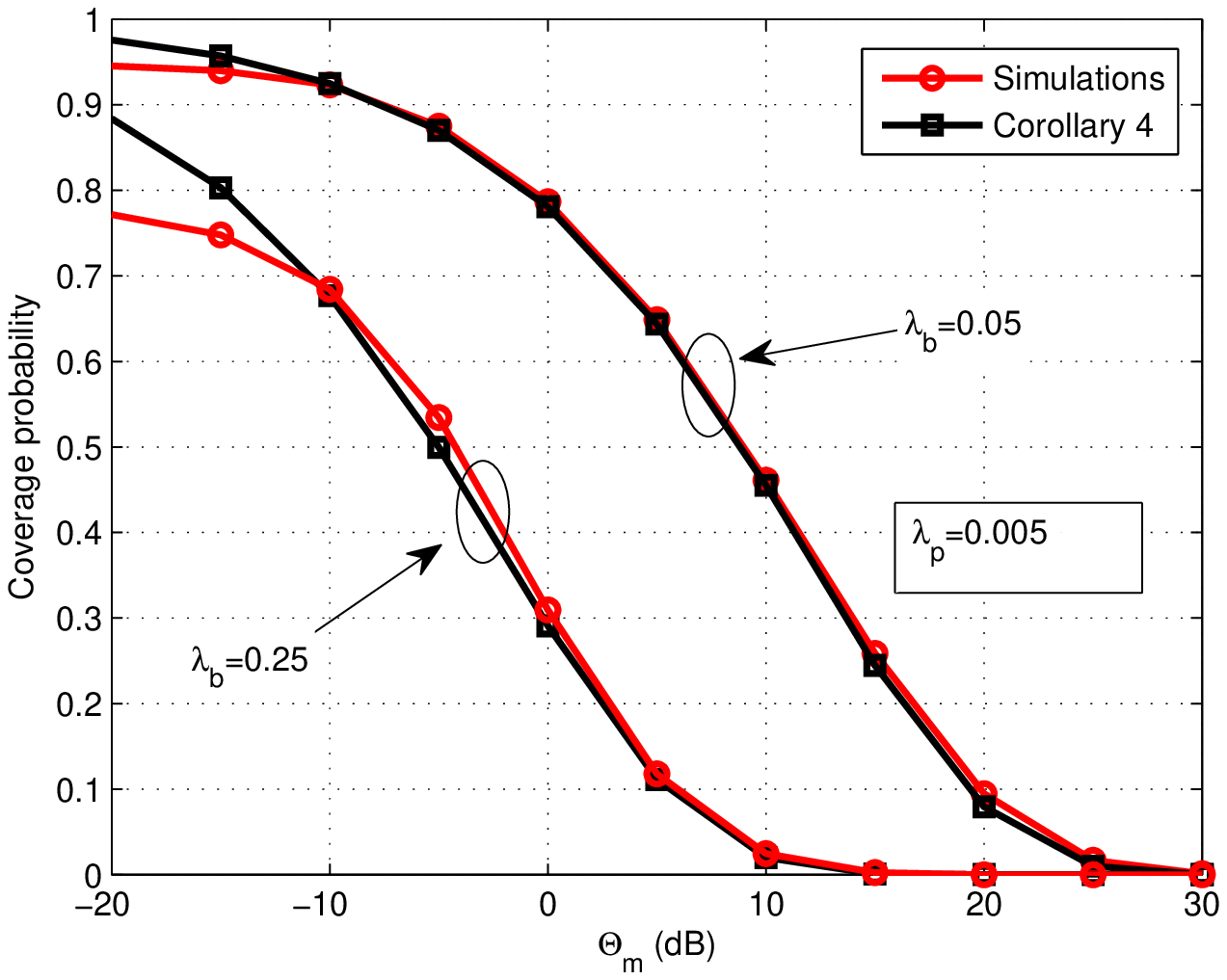}
		 \caption{Corollary 4}
		 \label{sub:fig_cor4}
	\end{subfigure}\hfill
\caption{Simulations versus analysis (Theorem 1 and Corollary 4)}
		\label{fig:sim_vs_anal_thm1_and_cor4}
\end{figure*}

Fig. \ref{fig:lemma1_2} illustrates the impact of the number of PBs  $N_p$ in the harvesting zone on the average received power  $\bar{P}_{BN}$. It can be noticed that, for Lemma 1, the average received power $\bar{P}_{BN}$ increases with an increase in $N_P$,   whereas, for Lemma 2, the $\bar{P}_{BN}$ is constant because the harvesting zone is the whole network in Lemma 2. The important observation  is that most of the power is harvested from the nearest PBs especially for $\alpha=4$. In case of $\alpha=4$, the contribution from the far away PBs is very small, therefore, $\bar{P}_{BN}$ quickly converges to the maximum value. On the other hand,   for $\alpha=3$, the contribution from those far away nodes is still significant and therefore convergence is not achieved.  
\begin{figure}
	\centering
		\includegraphics[scale=0.7]{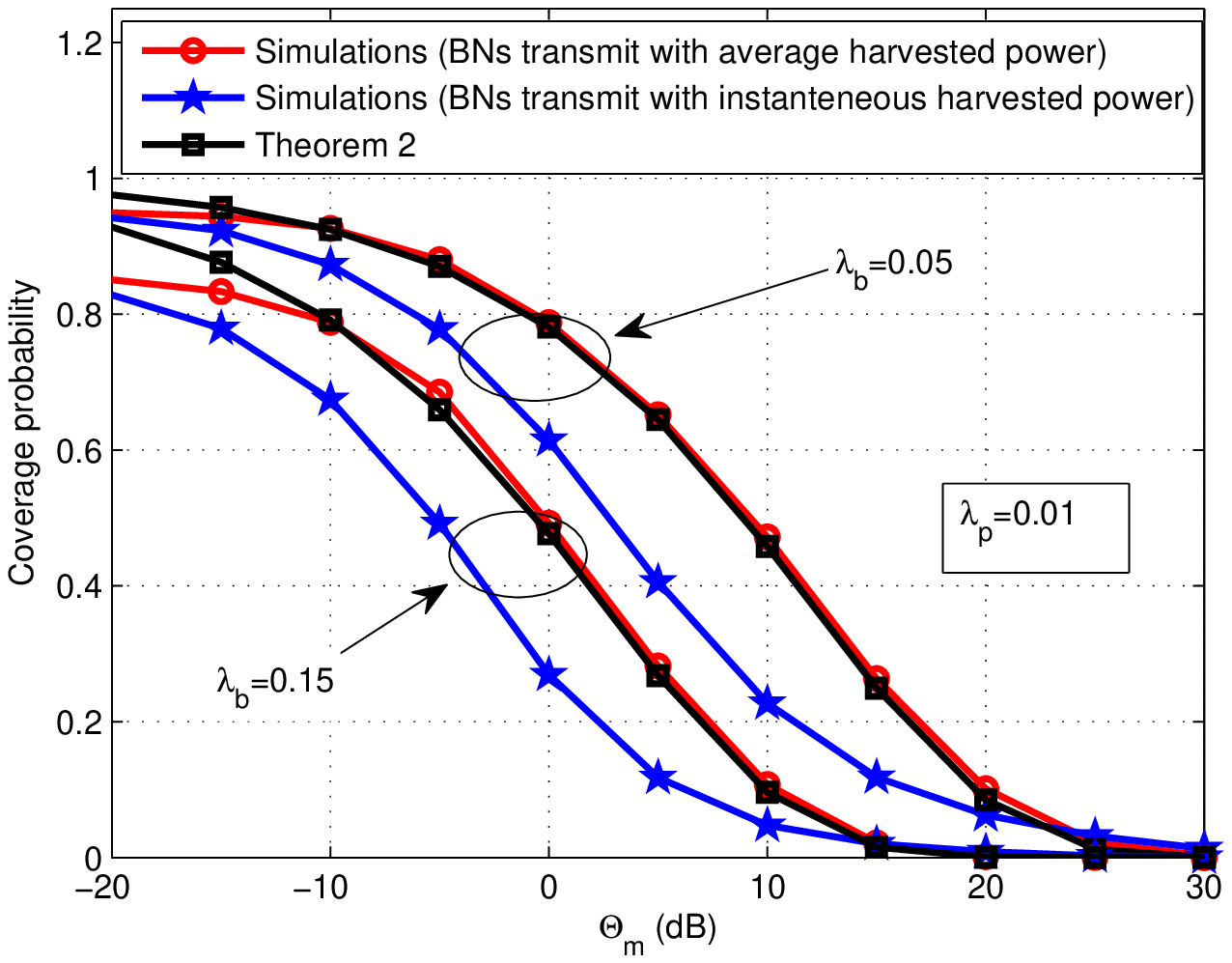}
		\caption{Simulations versus analysis (Theorem 2)}
		\label{fig_thm2}
\end{figure}

In Fig. \ref{fig:sim_vs_anal_thm1_and_cor4}, we compare the Monte Carlo simulations results with the numerical results for various network parameters. It can be noticed that the numerical and simulations curves are matching well for both  Theorem 1 and Corollary 4, which shows the accuracy of the analysis. It is important to remind that, in the proof of Theorem 1, we used a simplified expression for the SINR in order to make the analysis tractable. Similarly, in Corollary 4, we used the averaged received power $\bar{P}_{BN}$ as the harvested power at BN. 

Similar to Fig. \ref{fig:sim_vs_anal_thm1_and_cor4}, we compare the numerical and simulation results in Fig. \ref{fig_thm2}  for Theorem 2. It can be noticed that there are two different curves for the simulations. The first one is obtained when the BNs utilize the average harvested power (given in Lemma \ref{avg_power_2}), while the second is obtained when the BNs utilize the instantaneous harvested power. It can be noticed that the numerical curve is very close to the simulations' curve, which is obtained when the BNs transmit with the average harvested power. Whereas,  the second simulations curve provides the lower bound.  This illustrates that the random transmit power of the BNs degrades the coverage probability of the network. 

Fig. \ref{impact_den_PBs} shows the impact of the density of the PBs on the coverage probability for different SINR thresholds. It can be noticed from Fig. \ref{sub_fig_low_threshold} that the coverage probability increases with an increase in density of the PBs for different values of $\lambda_b$.
This improvement in the coverage probability is due to the increase in the received (and reflected) power of at the BNs as a result of the dense deployment of the PBs. 
However, for a large $\Theta_m$, the dense deployment of PBs can reduce the SINR coverage probability, as shown in Fig. \ref{sub_fig_high_threshold}. From Fig. \ref{sub_fig_high_threshold}, we observe that the coverage probability reaches to a maximum point and then starts decreasing with the increase in $\lambda_p$. 
This decrease in the coverage probability at higher SINR threshold is due to the increase in the interference power. The high SINR is achieved only by a very small fraction of the BNs, which are very close to the PBs. The distance between the PBs and BNs decreases with an increase in $\lambda_p$, which improves the harvested and reflected power of the BNs.
As a result the interference power also increases and the coverage probability drops.
This confirms our earlier observation in Remark \ref{remark2} and suggests that in order to achieve the optimal coverage, the density of the PBs should be adjusted by taking  $\Theta_m$ into account.
  
\begin{figure*} 
\centering
			\begin{subfigure}{.5\textwidth}
			\centering
			\includegraphics[width=1\linewidth]{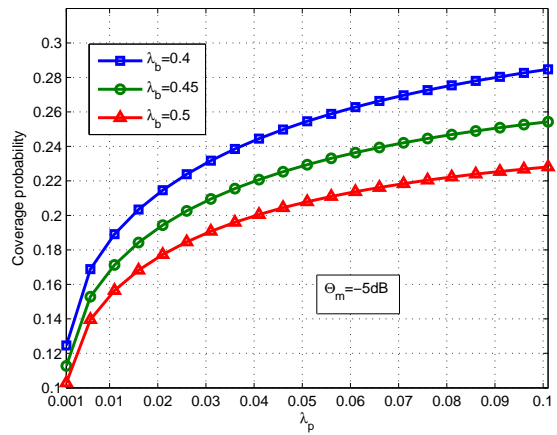}
			\caption{$\Theta_m=-5$dB}
			\label{sub_fig_low_threshold}
		\end{subfigure}\hfill \\
	\begin{subfigure}{.5\textwidth}
		\centering
		\includegraphics[width=1\linewidth]{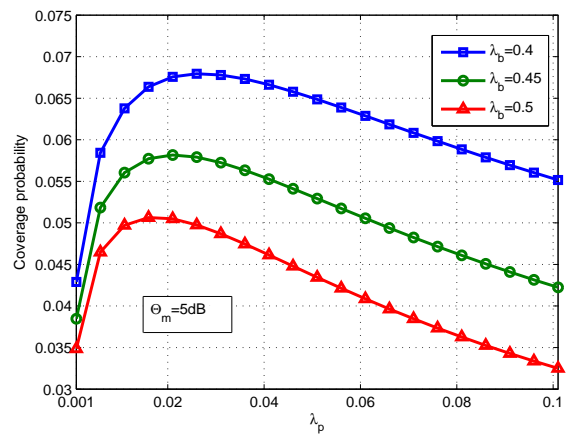}
		 \caption{$\Theta_m=5$dB}
		 \label{sub_fig_high_threshold}
	\end{subfigure}\hfill
\caption{Impact of the density of the PBs $\lambda_p$ on the coverage probability $\mathbb{P}_s$}
		\label{impact_den_PBs}
\end{figure*}

\begin{figure*}
		\begin{subfigure}{0.5\textwidth}
			\centering
			\includegraphics[width=1\linewidth]{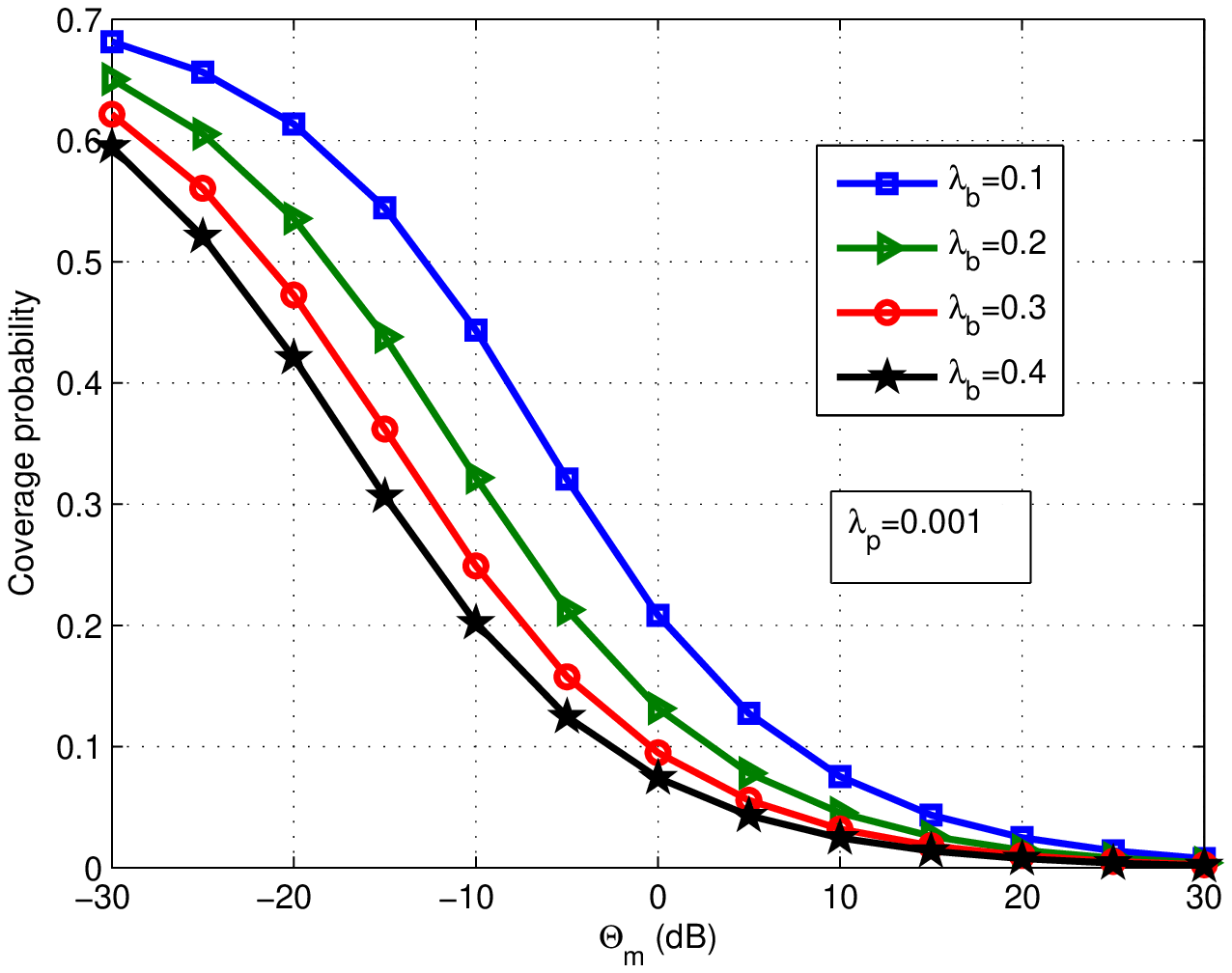}
			\caption{Coverage probability $\mathbb{P}_s$}
			\label{sub:fig_cov1}
		\end{subfigure}\hfill 
	\begin{subfigure}{0.5\textwidth}
	  \centering
		\includegraphics[width=1\linewidth]{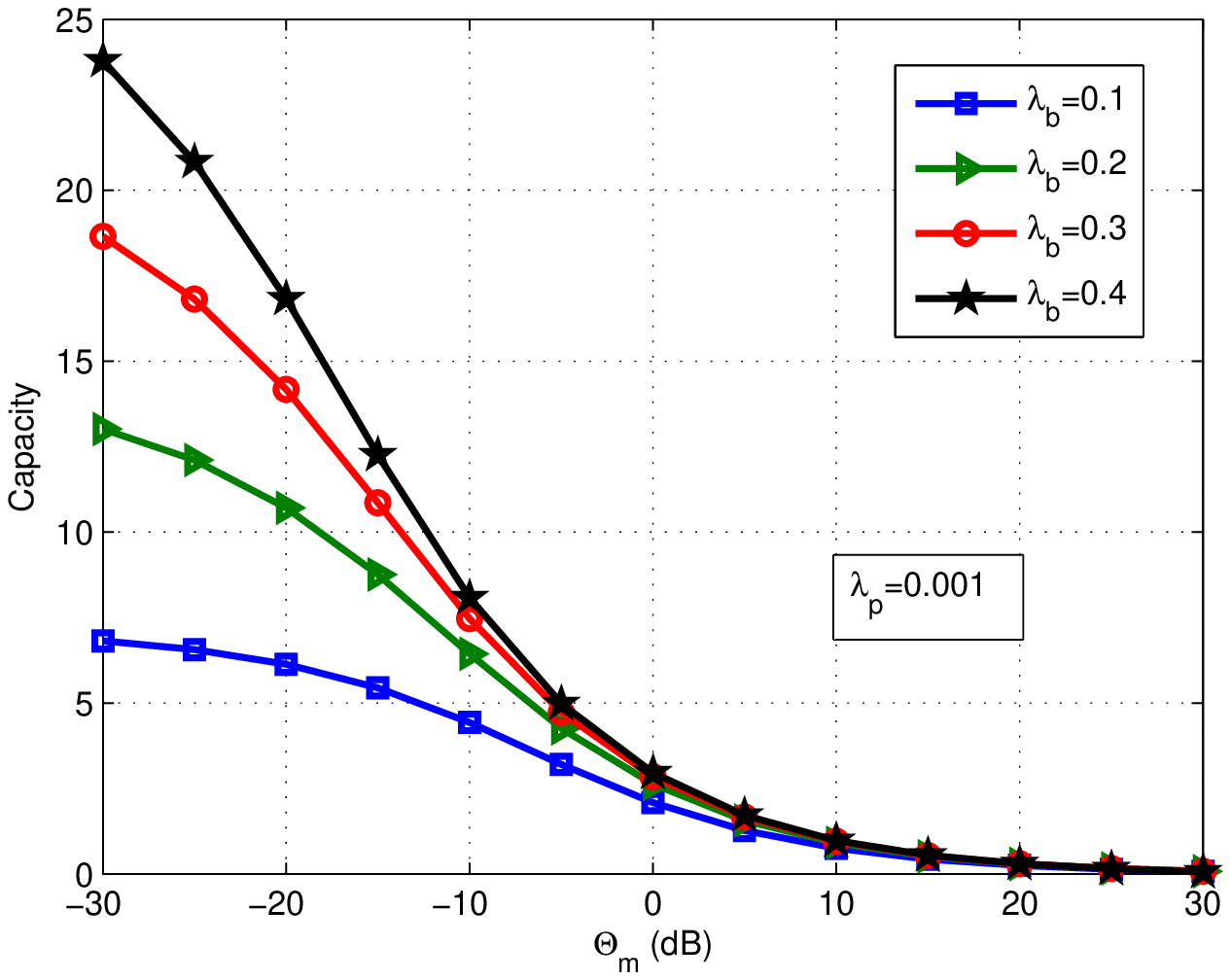}
		 \caption{Capacity $\mathbb{C}$}
		 \label{sub:fig_cap1}
	\end{subfigure}\hfill
	\begin{subfigure}{0.5\textwidth}
			\centering
			\includegraphics[width=1\linewidth]{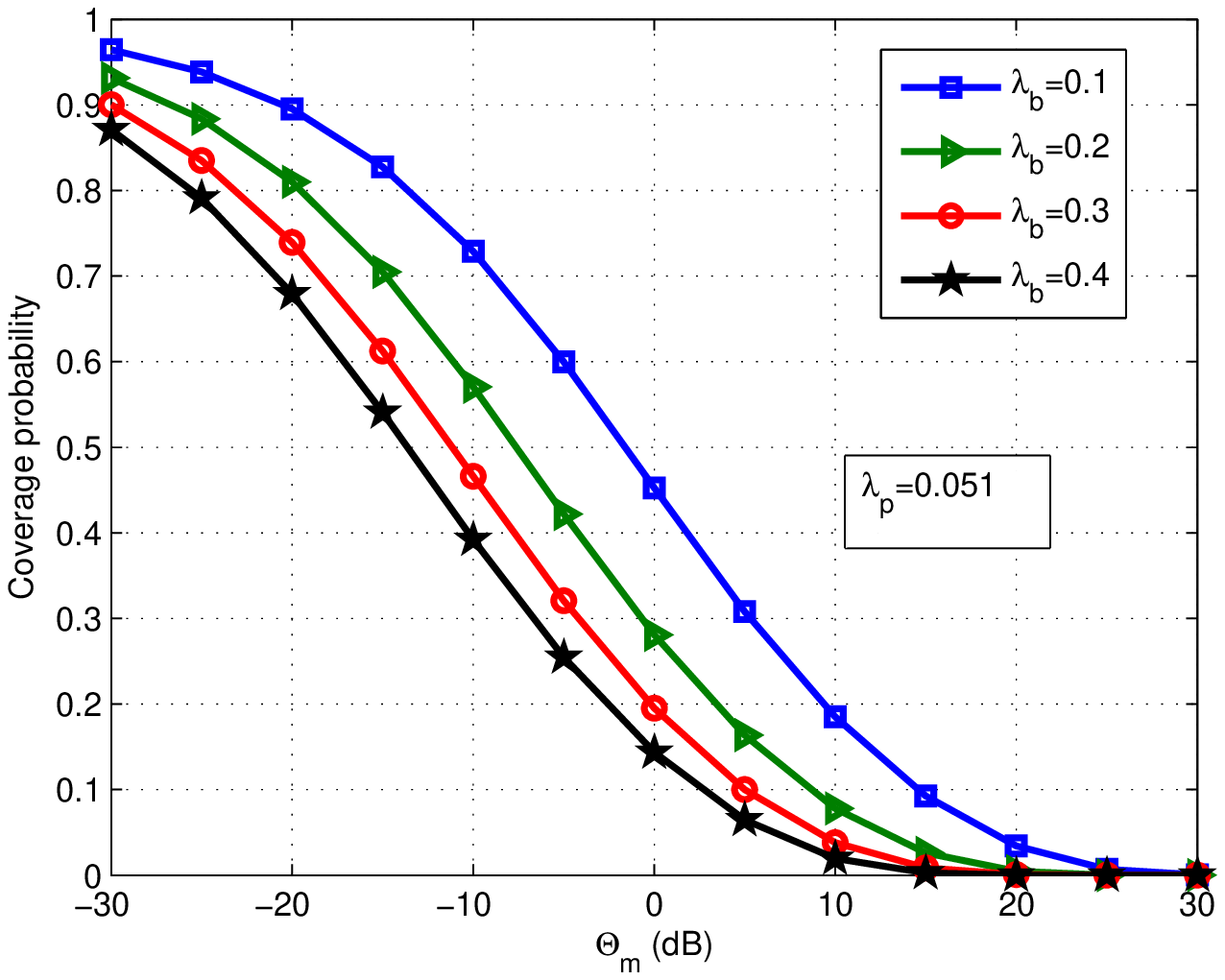}
			\caption{Coverage probability $\mathbb{P}_s$}
			\label{sub:fig_cov2}
		\end{subfigure} \hfill
	\begin{subfigure}{0.5\textwidth}
	  \centering
		 \includegraphics[width=1\linewidth]{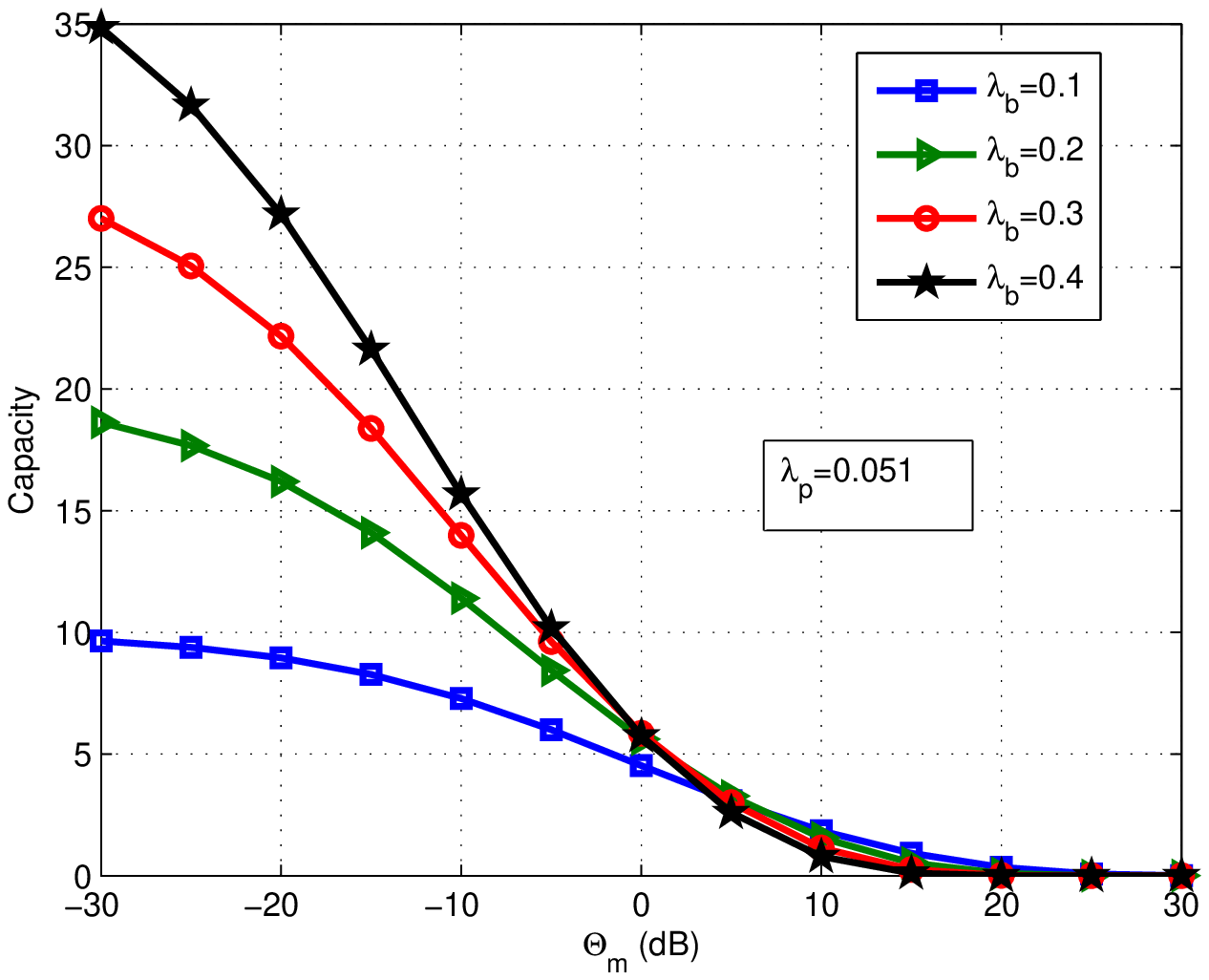}
		 \caption{Capacity $\mathbb{C}$}
		 \label{sub:fig_cap2}
	\end{subfigure}\hfill
	\caption{Impact of density of the BNs $\lambda_b$ on the coverage probability $\mathbb{P}_s$ and capacity $\mathbb{C}$ of the network}
\label{impact_den_BN}
\vspace{-1.7em}
\end{figure*}
Fig. \ref{impact_den_BN} depicts the impact of the density of the BNs on the coverage probability and capacity of the network. It can be observed from both Fig. \ref{sub:fig_cov1} and Fig. \ref{sub:fig_cov2} that an increase in the density of the BNs decreases the coverage probability, which is due to the increase in the number of interferers. However, a more appropriate metric in this scenario might be the capacity of network, which is the number of successful transmissions in some unit area. Therefore, we plot the capacity corresponding to Fig. \ref{sub:fig_cov1} and Fig. \ref{sub:fig_cov2} in Fig. \ref{sub:fig_cap1} and Fig. \ref{sub:fig_cap2}, respectively. In contrast to the coverage probability, both Fig. \ref{sub:fig_cap1} and Fig. \ref{sub:fig_cap2} show that the capacity of the network increases with increase in the density of the BNs. 
However, at higher SINR threshold in Fig. \ref{sub:fig_cap2}  the capacity decreases with an increase in the density of the BNs. The high SINR threshold is achieved by the BNs which are in close vicinity of the PBs.  The number of interfering BNs and the interference increase with an increase in $\lambda_b$ due to which the BNs can not achieve such a high $\Theta_m$ and therefore, the capacity drops. The insights from Fig. \ref{impact_den_BN} again validate our earlier observations in Remark \ref{remark3} and \ref{remark4}.    
 The coverage probability can be seen as an objective function of an individual user whereas the capacity is a network wide metric and the service provider would be more interested in it.  The network provider would like to maximize the capacity of the network given that certain minimum quality of service (coverage probability here) should be satisfied. Therefore, in order to get an optimal performance the tradeoff between the coverage and capacity should be considered along with $\Theta_m$.

\begin{figure*} 
\centering
			\begin{subfigure}{.5\textwidth}
			\centering
			\includegraphics[width=1\linewidth]{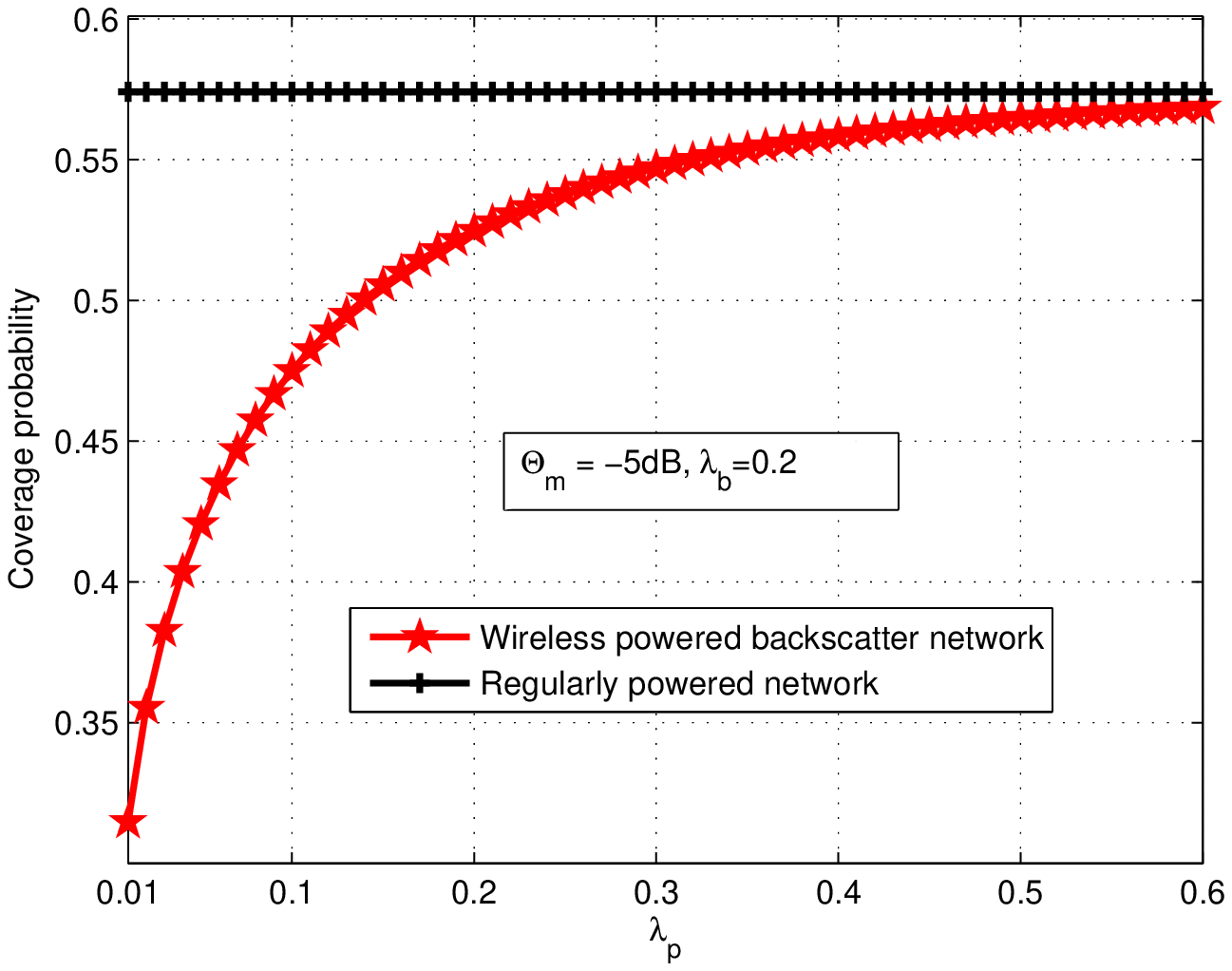}
			\caption{}
			\label{sub_low_theta}
		\end{subfigure}\hfill \\
	\begin{subfigure}{.5\textwidth}
		\centering
		\includegraphics[width=1\linewidth]{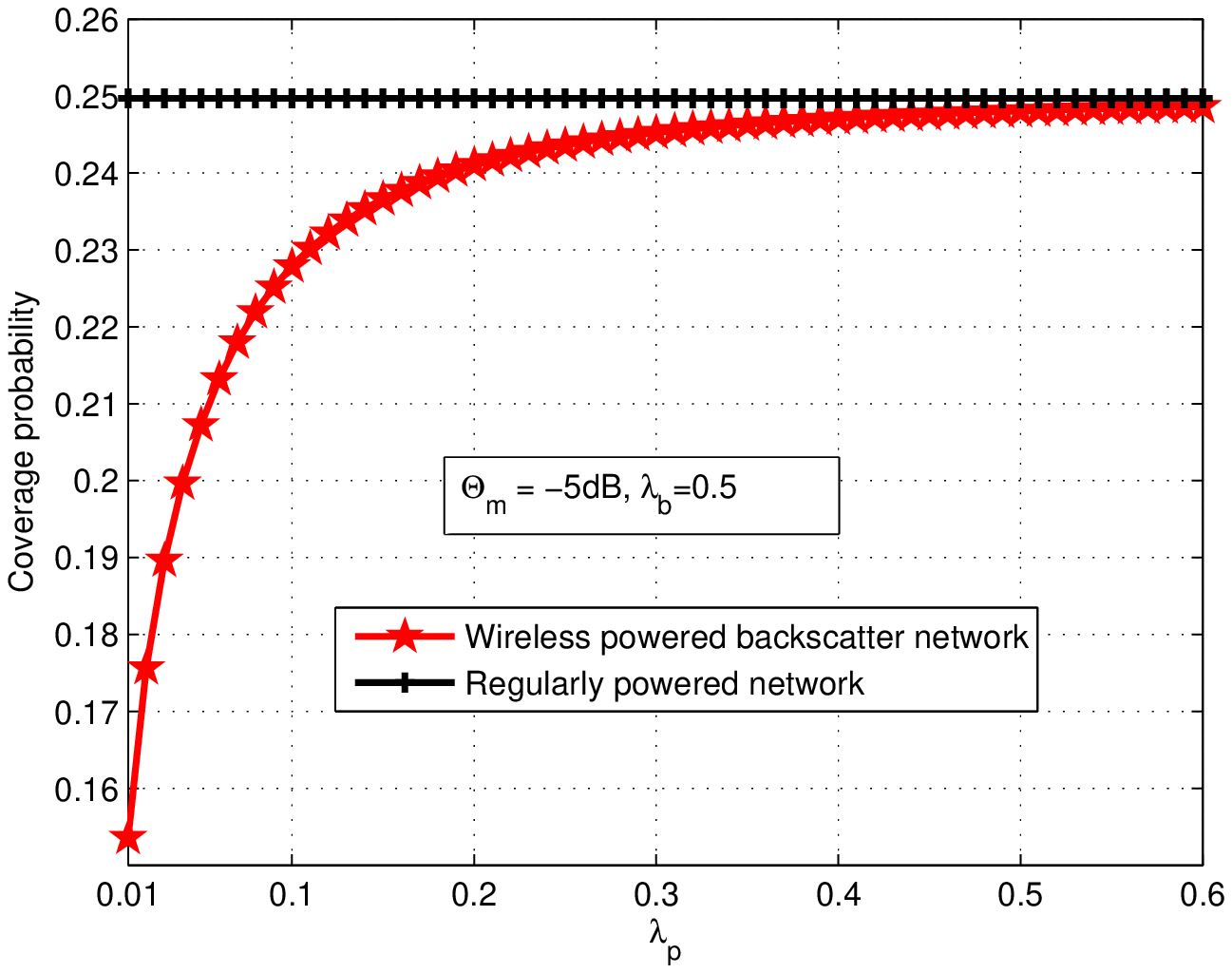}
		 \caption{}
		 \label{sub_high_theta}
	\end{subfigure}\hfill
\caption{Regularly powered vs. wireless powered backscatter network}
		\label{regular_vs_bcn}
\end{figure*}

In Fig. \ref{regular_vs_bcn} we compare the coverage of regularly powered network (RPN) with the wireless powered backscatter communications network (WPBN). In RPN there is no PB tier and each BN have constant transmit power equal to the transmit power of a PB i.e., $P_C$. Since, the transmit power of each BN is $P_C$, its coverage probability is the same for different values of $\lambda_p$ as shown in the figure. Whereas for the WPBN, the coverage probability increases with an increase in the density of the PBs. It can be noticed that when $\lambda_p=0.6$ then the coverage probability of the WPBN approaches to the coverage probability of the RPN. 
Another key point that can be observed from both Fig. \ref{sub_low_theta} and \ref{sub_high_theta} is that the coverage of WPBN increases very rapidly with increase in $\lambda_p$ initially, while after a certain increase in $\lambda_p$ the gain in coverage is very small. For instance, in both  Fig. \ref{sub_low_theta} and \ref{sub_high_theta} the increase in the coverage probability is very significant when  $\lambda_p$ goes from $0.01$ to $0.3$, whereas from $0.3$ to $0.6$ the improvement in the coverage is almost negligible. This suggests that after certain point further densification of the WPBN with more PBs would not improve the coverage significantly.
\section{Conclusion}
%
In this paper, a general framework to study the performance of wireless powered backscattered IoT network is presented. To make the analysis tractable, a dimensionality reduction technique has been applied. The coverage and the capacity of the network in the form of tractable and numerically computable expressions are derived with the help of stochastic geometry. These expressions can be further simplified for some relevant special cases. It has been noticed that the coverage of the network can either increase or decrease with the increase in the density of the PBs depending upon the SINR threshold.  Furthermore, the capacity of the network increases with an increase in the density of the BNs, whereas the coverage probability decreases. Finally, the coverage of the wireless powered backscatter network has been compared with the coverage of the regularly powered network and it has been observed that for very high density of the PBs the coverage of the wireless powered backscattered network approaches to the coverage of the regular powered network.

A future study might consider a minimum circuit-power threshold on the harvest power of the BN and if the harvested power is not satisfied then the BN will not backscatter the signal. It would also be interesting to incorporate different collisions resolution schemes such as using directional antennas and/or successive interference cancellation and study its effect on the performance of this wireless powered backscatter network. 
\appendices
\section{}
\textit{Proof of Theorem 1:}
By using the definition of $\gamma_{R_0}$ in \eqref{sinr_2nd}, we write $\mathbb{P}_s$ as
\begin{multline}
\mathbb{P}_s =   \mathbb{P}\left[\gamma_{R_0}\geq\Theta_{m}\right]= \idotsint\displaylimits_{0<x_{0,1}<x_{0,2}<\cdots<x_{0,N_p}<\infty}\mathbb{P}\left[\gamma_{R_0}\geq\Theta_{m} | x_{0,1},\cdots, x_{0,N_p} \right] \\ f_x(x_{0,1},\cdots,x_{0,N_p})\mathrm{d}x_{0,1},\cdots,\mathrm{d}x_{0,N_p},
\label{thm1_proof_eq1}
\end{multline}
where the integration is over the joint PDF of the distances $x_{0,1},\cdots,x_{0,N_p}$ to nearest $N_p$ PBs.
Now, the conditional probability $\mathbb{P}\left[\gamma_{R_0}\geq\Theta_{m} | x_{0,1},\cdots, x_{0,N_p} \right]$  can be written as
\begin{multline}
\mathbb{P}\left[\gamma_{R_0}\geq\Theta_{m} | x_{0,1},\cdots, x_{0,N_p} \right]=\mathbb{P}\left[\frac{ \textbf{h}_{0,0}^b d_{0,0}^{-\alpha^b}\sum_{i=1}^{N_p} x_{0,i}^{-\alpha^f}}{\sum_{j\in\Phi_b\backslash BN_0}\textbf{h}_{0,j}^b d_{0,j}^{-\alpha^b}\sum_{k=1}^{N_p} x_{j,k}^{-\alpha^f}+\frac{N_0}{\beta P_C}} \geq \Theta_{m} \right] \\
=\mathbb{P}\left[\textbf{h}_{0,0}^b \geq s \left(I+\frac{N_0}{\beta P_C}\right)\right]\overset{a}=\mathbb{E}_I\left[\exp(-s\left(I+\frac{N_0}{\beta P_C}\right))\right]= \exp(-\frac{s N_0}{\beta P_C})\mathcal{L}_I(s),
\label{thm1_proof_eq2}
\end{multline}
where $s=\frac{\Theta_{m}d_{0,0}^{\alpha^b}}{\sum_{i=1}^{N_p}x_{0,i}^{-\alpha^f}}$, $I = \sum_{j\in\Phi_b\backslash BN_0}\textbf{h}_{0,j}^b d_{0,j}^{-\alpha^b}\sum_{k=1}^{N_p}x_{j,k}^{-\alpha^f}$, $\left(a\right)$ follows because $\textbf{h}_{0,0}^b\sim \exp(1)$, and $\mathcal{L}_I(s)=\mathbb{E}_I\left[\exp\left(-sI\right)\right]$ is the Laplace transform of the interference $I$ evaluated at $s$.  Next, we find the Laplace transform as
\begin{multline}
\mathcal{L}_I(s) =  \mathbb{E}_I\left[\exp(-sI)\right] = \mathbb{E}_{\textbf{h}_{0,j}^b,d_{0,j},x_{j,1},\cdots,x_{j,N_p}}\left[\exp\left(-s\sum_{j\in\Phi_b\backslash BN_0}\textbf{h}_{0,j}^b d_{0,j}^{-\alpha^b}\sum_{k=1}^{N_p}x_{j,k}^{-\alpha^f}\right)\right] \\
\overset{a} = \mathbb{E}_{d_{0,j},x_{j,1},\cdots,x_{j,N_p}}\left[\prod_{j\in\Phi_b\backslash BN_0}\mathbb{E}_{\textbf{h}_{0,j}^b}\left[\exp\left(-s \textbf{h}_{0,j}^b d_{0,j}^{-\alpha^b}\sum_{k=1}^{N_p}x_{j,k}^{-\alpha^f} \right)\right]\right] \\
\overset{b}= \mathbb{E}_{d_{0,j}}\left[\prod_{j\in\Phi_b\backslash BN_0}\mathbb{E}_{x_{j,1},\cdots,x_{j,N_p}}\left[\frac{1}{1+s d_{0,j}^{-\alpha^b}\sum_{k=1}^{N_p}x_{j,k}^{-\alpha^f}}\right]\right] \\
\overset{c}= \exp\left(-2\pi \lambda_b \int_{0}^{\infty} \mathbb{E}_{x_1,\cdots,x_{N_p}} \left[\frac{1}{1+\frac{u^{\alpha^b}}{s\sum_{k=1}^{N_p} x_{k}^{-\alpha^f}}}\right]u\mathrm{d}u\right), \\
 \overset{d}=\exp\left(-\frac{2\pi^2 \lambda_b s^{2/\alpha^b}}{\alpha^b \sin\left(\frac{2 \pi}{\alpha^b}\right)} \mathbb{E}_{x_1,\cdots,x_{N_p}}\left(\sum_{k=1}^{N_p} x_k^{-\alpha^f}\right)^{2/\alpha^b}\right)
\label{lap_thm1}
\end{multline}
where $\left(a\right)$ follows due to the independence of $\textbf{h}_{0,j}^b$, $\left(b\right)$ follows from the fact that $\textbf{h}_{0,j}^b\sim \exp(1)$ and the independence of $\mathbb{E}_{x_{j,1},\cdots,x_{j,N_p}}$,  $\left(c\right)$ follows from the probability generating functional (PGFL) of 2-D PPP \cite{martin_book}, changing the coordinates to the polar coordinates and evaluating one of the integral, and $\left(d\right)$ follows by changing the order of integration and expectation and then computing the integral. 
The last expectation $\mathbb{E}_{x_1,\cdots, x_{N_p}}$ can be computed by using the joint PDF given in \eqref{joint_pdf_dist}. By plugging \eqref{lap_thm1} in \eqref{thm1_proof_eq2}, and then \eqref{thm1_proof_eq2} in \eqref{thm1_proof_eq1}, we get the final expression for $\mathbb{P}_s$.
\section{}
\textit{Proof of Theorem 2:}
The SIR coverage probability can be written as
\begin{equation}
\mathbb{P}_s = \mathbb{P}\left[\gamma_{R_0}\geq\Theta_{m}\right]= \mathbb{P}\left[\frac{ \textbf{h}_{0,0}^b d_{0,0}^{-\alpha^b}\sum_{i\in\Phi_p} x_{0,i}^{-\alpha^f}}{\sum_{j\in\Phi_b\backslash BN_0}\textbf{h}_{0,j}^b d_{0,j}^{-\alpha^b}\sum_{k\in\Phi_p}x_{j,k}^{-\alpha^f}}\geq \Theta_m\right],
\label{SIR_covC_1}
\end{equation}
where we put \eqref{sir_second} for $\gamma_{R_0}$. We know that $\textbf{h}_{0,0}^b$ is exponentially distributed with unit mean due to which the above expression becomes
\begin{equation}
\mathbb{P}_s =\mathbb{E}_{\Phi_b, \Phi_p, \textbf{h}_{0,j}^b}\left[\exp\left(-\frac{\Theta_m\sum_{j\in\Phi_b\backslash BN_0}\textbf{h}_{0,j}^b d_{0,j}^{-\alpha^b}\sum_{k\in\Phi_p}x_{j,k}^{-\alpha^f}}{d_{0,0}^{-\alpha^b}\sum_{i\in\Phi_p} x_{0,i}^{-\alpha^f}}\right)\right],
\label{SIR_covC_2}
\end{equation}
where the expectation should be with respect to the $\Phi_b, \Phi_p$, and $\textbf{h}_{0,j}^b$. Due to the independence of $\textbf{h}_{0,j}^b$ the previous expression can be written as
\begin{multline}
\mathbb{P}_s =\mathbb{E}_{\Phi_b, \Phi_p}\left[\prod_{j\in\Phi_b\backslash BN_0}\mathbb{E}_{\textbf{h}_{0,j}^b}\left[\exp\left(-\frac{\Theta_m \textbf{h}_{0,j}^b d_{0,j}^{-\alpha^b}\sum_{k\in\Phi_p}x_{j,k}^{-\alpha^f}}{d_{0,0}^{-\alpha^b}\sum_{i\in\Phi_p} x_{0,i}^{-\alpha^f}}\right)\right]\right]\\
\overset{a}=\mathbb{E}_{\Phi_b}\left[\prod_{j\in\Phi_b\backslash BN_0}\mathbb{E}_{\Phi_p}\left[\frac{1}{1+\frac{\Theta_m d_{0,j}^{-\alpha^b}\sum_{k\in\Phi_p}x_{j,k}^{-\alpha^f}}{d_{0,0}^{-\alpha^b}\sum_{i\in\Phi_p} x_{0,i}^{-\alpha^f}}}\right]\right]\\
\overset{b}=\exp\left(-2\pi\lambda_b\int_0^{\infty}\left(1-\mathbb{E}_{\Phi_p}\left[\frac{1}{1+\frac{\Theta_m u^{-\alpha^b}\sum_{k\in\Phi_p}x_{j,k}^{-\alpha^f}}{d_{0,0}^{-\alpha^b}\sum_{i\in\Phi_p} x_{0,i}^{-\alpha^f}}}\right]\right)u\mathrm{d}u\right)\\
\overset{c}\geq\exp\left(-2\pi\lambda_b\int_0^{\infty}\left(1-\frac{1}{1+\frac{\Theta_m u^{-\alpha^b}}{d_{0,0}^{-\alpha^b}}\mathbb{E}_{\Phi_p}\left[\frac{\sum_{k\in\Phi_p}x_{j,k}^{-\alpha^f}}{\sum_{i\in\Phi_p} x_{0,i}^{-\alpha^f}}\right]}\right)u\mathrm{d}u\right),
\label{sir_cov_2nd}
\end{multline}
where $\left(a\right)$ follows because $\textbf{h}_{0,j}^b\sim \exp(1)$ and both $\Phi_b$ and $\Phi_p$ are independent,  $\left(b\right)$ follows due the probability generating functional of 2-D PPP \cite{martin_book}, while in $\left(c\right)$ we utilize the Jensen's inequality. Now to find the last expectation $\mathbb{E}_{\Phi_p}$, we use $\mathbb{E}_{\Phi_p}\left[\frac{\sum_{k\in\Phi_p}x_{j,k}^{-\alpha^f}}{\sum_{i\in\Phi_p} x_{0,i}^{-\alpha^f}}\right]\geq\frac{\mathbb{E}_{\Phi_p}\left[\sum_{k\in\Phi_p}x_{j,k}^{-\alpha^f}\right]}{\mathbb{E}_{\Phi_p}\left[\sum_{i\in\Phi_p} x_{0,i}^{-\alpha^f}\right]}=1$, where to find $\mathbb{E}_{\Phi_p}$ we utilize Lemma 2.  The final expression in \eqref{eq:thm2} is obtained by doing some algebra. 
\section*{Acknowledgment}
We thank Kaifeng Han and Kaibin Huang for discussion and participation in the initial stage of this work.
\ifCLASSOPTIONcaptionsoff
  \newpage
\fi
\renewcommand\refname{References} 
\bibliographystyle{IEEEtran}
\bibliography{IEEEabrv,Reference}


%
%
%
%
%
\end{document}